\documentclass[11pt]{article}%,12pt,onecolumn,draftclsnofoot,

\usepackage{amsthm,amsmath,amssymb,enumerate,algorithm,pgfplots,graphicx,multirow,subcaption,cite}
\RequirePackage[colorlinks,citecolor=blue,urlcolor=blue]{hyperref}
\usetikzlibrary{positioning,arrows.meta,decorations.text}
\usetikzlibrary{shapes.multipart}

\usepackage[margin=1in]{geometry}
\usepackage[titletoc,toc,title]{appendix}

\theoremstyle{plain}
\newtheorem{theorem}{Theorem}[section]
\newtheorem{proposition}[theorem]{Proposition}
\newtheorem{definition}[theorem]{Definition}
\newtheorem{lemma}[theorem]{Lemma}
\newtheorem{example}[theorem]{Example}
\newtheorem{corollary}[theorem]{Corollary}
\newtheorem{obs}[theorem]{Observation}

\newcommand{\R}{\mathbb{R}}
\newcommand{\sub}{\subseteq}
\newcommand{\sgn}{\operatorname{sgn}}

\newcommand{\ds}{\displaystyle}

\newcommand{\U}{\mathcal{U}}

\DeclareMathOperator*{\argmin}{arg\,min}

\newcommand{\st}{\operatorname{s.t.}}
\newcommand{\prox}{\operatorname{prox}}
\newcommand{\soft}{\operatorname{soft}}
\newcommand{\sd}{\partial}

\newcommand{\cg}{\succcurlyeq}
\newcommand{\psd}{\cg\mb0}

\newcommand{\diag}{\operatorname{diag}}
\newcommand{\rank}{\operatorname{rank}}

\newcommand{\bs}{\boldsymbol}
\newcommand{\mb}[1]{\mathbf{#1}}

\newcommand{\be}{\bs\epsilon}

\newcommand{\bb}{{\bs\beta}}

\newcommand{\bd}{\bs\delta}
\newcommand{\ba}{\bs\alpha}

\newcommand{\ph}{\bs \phi}
\newcommand{\bg}{\bs\gamma}

\newcommand{\D}{\bs\Delta}
\newcommand{\s}{\bs\Sigma}
\newcommand{\Ph}{\bs\Phi}
\newcommand{\T}{\bs\Theta}

\newcommand{\barl}{\overline{\lambda}}
\newcommand{\ubarl}{\underline{\lambda}}
\newcommand{\w}{\mb w}
\newcommand{\x}{\mb x}

\newcommand{\zz}{\mb z}

\newcommand{\y}{\mb y}

\newcommand{\A}{\mb A}
\newcommand{\X}{{\mb X}}

\newcommand{\ns}{1000} % for figures (number of knots)
\newcommand{\cl}{(\textsc{CL}_{\mu,\gamma})}
\newcommand{\tl}{(\textsc{TL}_{\lambda,\ell})}
\newcommand{\tla}[1]{(\textsc{TL}_{#1})}
\newcommand{\cla}[1]{(\textsc{CL}_{#1})}
\newcommand{\I}{\mb I}

\newcommand{\tk}[1]{T_k\left(#1\right)}
\newcommand{\tka}[2]{T_{#1}\left(#2\right)}
\newcommand{\tke}{T_k}
\newcommand{\N}{\mathcal{N}}

\newcommand{\cL}{\mathcal{L}}
\newcommand{\whbb}{\widehat{\bb}}
\renewcommand\[{\begin{equation}}
\renewcommand\]{\end{equation}}
\newcommand{\bbs}{\bb^*}

\begin{document}

\title{The Trimmed Lasso: Sparsity and Robustness}

\author{Dimitris~Bertsimas,
        Martin~S.~Copenhaver,        and~Rahul~Mazumder\thanks{Authors' affiliation: Sloan School of Management and Operations Research Center, MIT. 
\newline
Emails: \{dbertsim,mcopen,rahulmaz\}@mit.edu.}
}

\date{August 15, 2017}

\maketitle

\begin{abstract}
Nonconvex penalty methods for sparse modeling in linear regression have been a topic of fervent interest in recent years.  
Herein, we study a family of nonconvex penalty functions that we call the trimmed Lasso and that offers exact control over the desired level of sparsity of estimators. We analyze its structural properties and in doing so show the following:
\begin{enumerate}
\item Drawing parallels between robust statistics and robust optimization, we show 
that the trimmed-Lasso-regularized least squares problem can be viewed as a generalized form of total least squares under a specific model of uncertainty. In contrast, 
this same model of uncertainty, viewed instead through a robust optimization lens, leads to the convex SLOPE (or OWL) penalty.

\item Further, in relating the trimmed Lasso to commonly used sparsity-inducing penalty functions, we provide a succinct characterization of the connection between trimmed-Lasso-like approaches and penalty functions that are coordinate-wise separable, showing that the trimmed penalties subsume existing coordinate-wise separable penalties, with strict containment in general.

\item Finally, we describe a variety of exact and heuristic algorithms, both existing and new, for trimmed Lasso regularized estimation problems. 
We include a comparison between the different approaches and an accompanying implementation of the algorithms.

\end{enumerate}

\end{abstract}

%%%%%%%%%%%%%%%%%%%%%%%
%%% Introduction
%%%%%%%%%%%%%%%%%%%%%%%

\section{Introduction}\label{sec:intro}

Sparse modeling in linear regression has been a topic of fervent interest in recent years \cite{hastie,bvdg}. This interest has taken several forms, from substantial developments in the theory of the Lasso to advances in algorithms for convex optimization. Throughout there has been a strong emphasis on the increasingly high-dimensional nature of linear regression problems; in such problems, where the number of variables $p$ can vastly exceed the number of observations $n$, sparse modeling techniques are critical for performing inference.

\subsubsection*{Context}

One of the fundamental approaches to sparse modeling in the usual linear regression model of $\y=\X\bb+\be$, with $\y\in\R^n$ and $\X\in\R^{n\times p}$, is the
best subset selection \cite{miller2002subset} problem:
\begin{equation}\label{eqn:a10}
\ds\min_{\|\bb\|_0\leq k}\frac{1}{2} \|\y-\X\bb\|_2^2,
\end{equation}
which seeks to find the best choice of $k$ from among $p$ features that best explain the response in terms of the least squares loss function.
%%%features from a dataset with $n$ samples and $p$ features.
The problem \eqref{eqn:a10} has received extensive attention from a variety of statistical and optimization perspectives---see for example \cite{bkm} and references therein. One can also consider the Lagrangian, or penalized, form of \eqref{eqn:a10}, namely,
\begin{equation}\label{eqn:a10p}
\ds\min_{\bb}\frac{1}{2} \|\y-\X\bb\|_2^2 + \mu \|\bb\|_0,
\end{equation}
for a regularization parameter $\mu>0$. One of the advantages of \eqref{eqn:a10} over \eqref{eqn:a10p} is that it offers direct control over estimators' sparsity via the discrete parameter $k$, as opposed to the Lagrangian form \eqref{eqn:a10p} for which the correspondence between the continuous parameter $\mu$ and the resulting sparsity of estimators obtained is not entirely clear. For further discussion, see \cite{convspen}.

Another class of problems that have received considerable attention in the statistics and machine learning literature is the following:
\begin{equation}\label{eqn:a1}
\min_\bb \frac{1}{2}\|\y-\X\bb\|_2^2 + R(\bb),
\end{equation}
where $R(\bb)$ is a choice of regularizer which encourages sparsity in $\bb$.
For example, the popularly used Lasso \cite{tibshirani} takes the form of problem \eqref{eqn:a1} with $R(\bb)=\mu\|\bb\|_1$, where $\|\cdot\|_1$ is the $\ell_1$ norm; in doing so, the Lasso simultaneously selects variables and also performs shrinkage. 
The Lasso has seen widespread success across a variety of applications.

In contrast to the convex approach of the Lasso, there also has been  been growing interest in considering richer classes of regularizers $R$ which include nonconvex functions. Examples of such penalties include the $\ell_{q}$-penalty (for $q\in [0,1]$), minimax concave penalty (MCP) \cite{mcp}, and the smoothly clipped absolute deviation (SCAD) \cite{scad}, among others. Many of the nonconvex penalty functions considered are \emph{coordinate-wise separable}; in other words, $R$ can be decomposed as
$$R(\bb) = \sum_{i=1}^p \rho(|\beta_i|),$$
where $\rho(\cdot)$ is a real-valued function \cite{zhangzhang}. There has been a variety of evidence suggesting the promise of such nonconvex approaches in overcoming certain shortcomings of Lasso-like approaches.

One of the central ideas of nonconvex penalty methods used in sparse modeling is that of creating a continuum of estimation problems which bridge the gap between convex methods for sparse estimation (such as Lasso) and subset selection in the form \eqref{eqn:a10}. However, as noted above, such a connection does not necessarily offer direct control over the desired level of sparsity of estimators.

\subsection*{The trimmed Lasso}

In contrast with coordinate-wise separable penalties as considered above, we consider a family of penalties that are not separable across coordinates. One such penalty which forms a principal object of our study herein is 
$$\tk{\bb} := \min_{\substack{\|\ph\|_0\leq k}} \|\ph-\bb\|_1.$$
The penalty $\tke$ is a measure of the distance from the set of $k$-sparse estimators as measured via the $\ell_1$ norm. In other words, when used in problem \eqref{eqn:a1}, the penalty $R=\tke$ controls the amount of shrinkage towards sparse models. 

The penalty $\tke$ can equivalently be written as 
$$\tk{\bb} = \sum_{i=k+1}^p |\beta_{(i)}|,$$
where $|\beta_{(1)}|\geq |\beta_{(2)}|\geq \cdots\geq |\beta_{(p)}|$ are the sorted entries of $\bb$. In words, $\tk{\bb}$ is the sum of the absolute values of the $p-k$ smallest magnitude entries of $\bb$. The penalty was first introduced in \cite{thiao,hempel,gotoh1,gotoh2}. We refer to this family of penalty functions (over choices of $k$) as the \emph{trimmed Lasso}.\footnote{The choice of name is our own and is motivated by the least trimmed squares regression estimator, described below} The case of $k=0$ recovers the usual Lasso, as one would suspect. The distinction, of course, is that for general $k$, $\tke$ no longer shrinks, or biases towards zero, the $k$ largest entries of $\bb$.

Let us consider the least squares loss regularized via the trimmed lasso penalty---this leads to the following optimization criterion:
\begin{equation}\label{eqn:rmaux1}
\ds\min_{\bb} \frac{1}{2}\|\y-\X\bb\|_2^2 + \lambda \tk{\bb},
\end{equation}
where $\lambda>0$ is the regularization parameter. The penalty term shrinks the smallest $p-k$ entries of $\bb$ and does not impose any penalty on the largest $k$ entries of $\bb$. If $\lambda$ becomes larger, the smallest $p-k$ entries of $\bb$ are shrunk further; after a certain threshold---as soon as $\lambda \geq \lambda_0$ for some finite $\lambda_0$---the smallest $p-k$ entries are set to zero.
The existence of a finite $\lambda_0$ (as stated above) is an attractive feature of the trimmed Lasso and is known as its \emph{exactness} property, namely, for $\lambda$ sufficiently large, the problem \eqref{eqn:rmaux1} exactly solves constrained best subset selection as in problem \eqref{eqn:a10} (\emph{c.f.} \cite{gotoh1}). Note here the contrast with the separable penalty functions which correspond instead with problem \eqref{eqn:a10p}; as such, the trimmed Lasso is distinctive in that it offers precise control over the desired level of sparsity vis-\`a-vis the discrete parameter $k$. Further, it is also notable that many algorithms developed for separable-penalty estimation problems can be directly adapted for the trimmed Lasso.

Our objective in studying the trimmed Lasso is distinctive from previous approaches. In particular, while previous work on the penalty $\tke$ has focused primarily on its use as a tool for reformulating sparse optimization problems \cite{thiao,hempel} and on how such reformulations can be solved computationally \cite{gotoh1,gotoh2}, we instead aim to explore the trimmed Lasso's structural properties and its relation to existing sparse modeling techniques.

In particular,  a natural question we seek to explore is, what is the connection of the trimmed Lasso penalty with existing separable penalties commonly used in sparse statistical learning? For example, the trimmed Lasso bears a close resemblance to the clipped (or capped) Lasso penalty \cite{cl}, namely,
$$\sum_{i=1}^p \mu\min\{\gamma|\beta_i|,1\},$$
where $\mu,\gamma>0$ are parameters (when $\gamma$ is large, the clipped Lasso approximates $\mu\|\bb\|_0$).

\subsection*{Robustness: robust statistics and robust optimization}

A significant thread woven throughout the consideration of penalty methods for sparse modeling is the notion of robustness---in short, the ability of a method to perform in the face of noise. Not surprisingly, the notion of robustness has myriad distinct meanings depending on the context. Indeed, as Huber, a pioneer in the area of robust statistics, aptly noted:
\begin{quote}
``The word `robust' is loaded with many---sometimes inconsistent---connotations.'' \cite[p. 2]{huber}
\end{quote}
For this reason, we consider robustness from several perspectives---both the robust statistics~\cite{huber} and robust optimization~\cite{RObook} viewpoints.

A common premise of the various approaches is as follows: that a robust model should perform well even under small deviations from its underlying assumptions; and that to achieve such behavior, some efficiency under the assumed model should be sacrificed. Not surprisingly in light of Huber's prescient observation, the exact manifestation of this idea can take many different forms, even if the initial premise is ostensibly the same.

\subsubsection*{Robust statistics and the ``min-min'' approach}

One such approach is in the field of robust statistics \cite{huber,robRegBook,robStatsSurvey}. In this context, the primary assumptions are often probabilistic, i.e. distributional, in nature, and the deviations to be ``protected against'' include possibly gross, or arbitrarily bad, errors. Put simply, robust statistics is primary focused on analyzing and mitigating the influence of outliers on estimation methods.

There have been a variety of proposals of different estimators to achieve this. One that is particularly relevant for our purposes is that of \emph{least trimmed squares} (``LTS'') \cite{robRegBook}. For fixed $j\in\{1,\ldots,n\}$, the LTS problem is defined as
\begin{equation}\label{eqn:introlts}
\min_\bb \sum_{i=j+1}^n |r_{(i)}(\bb)|^2,
\end{equation}
where $r_i(\bb) = y_i-\x_i'\bb$ are the residuals and $r_{(i)}(\bb)$ are the sorted residuals given $\bb$ with $|r_{(1)}(\bb)|\geq |r_{(2)}(\bb)|\geq\cdots\geq |r_{(n)}(\bb)|$. In words, the LTS estimator performs ordinary least squares on the $n-j$ smallest residuals (discarding the $j$ largest or worst residuals).

Furthermore, it is particularly instructive to express \eqref{eqn:introlts} in the equivalent form (\emph{c.f.} \cite{bmlqs})
\begin{equation}\label{eqn:introltsalt}
\min_\bb\min_{\substack{I\sub \{1,\ldots,n\}:\\|I|=n-j }} \sum_{i\in I} |r_i (\bb)|^2.
\end{equation}
In light of this representation, we refer to LTS as a form of ``min-min'' robustness. One could also interpret this min-min robustness as \emph{optimistic} in the sense the estimation problems \eqref{eqn:introltsalt} and, \emph{a fortiori}, \eqref{eqn:introlts} allow the modeler to also choose observations to discard.

\subsubsection*{Other min-min models of robustness}

Another approach to robustness which also takes a min-min form like LTS is the classical technique known as \emph{total least squares} \cite{tls,tlsoverview}. For our purposes, we consider total least squares in the form
\begin{equation}\label{eqn:introeiv}
\min_{\bb}\min_{\D} \frac{1}{2}\|\y-(\X+\D)\bb\|_2^2 + \eta\|\D\|_{2}^2,
\end{equation}
where $\|\D\|_2$ is the usual Frobenius norm of the matrix $\D$ and $\eta>0$ is a scalar parameter. In this framework, one again has an optimistic view on error: find the best possible ``correction'' of the data matrix $\X$ as $\X+\D^*$ and perform least squares using this corrected data (with $\eta$ controlling the flexibility in choice of $\D$).

In contrast with the penalized form of \eqref{eqn:introeiv}, one could also consider the problem in a constrained form such as
\begin{equation}\label{eqn:introeivcon}
\min_\bb \min_{\D\in\mathcal{V}} \frac{1}{2}\|\y-(\X+\D)\bb\|_2^2,
\end{equation}
where $\mathcal{V}\sub\R^{n\times p}$ is defined as $\mathcal{V}= \{\D: \|\D\|_2\leq \eta'\}$ for some $\eta'>0$. 
This problem again has the min-min form, although now with perturbations $\D$ as restricted to the set $\mathcal{V}$.

\subsubsection*{Robust optimization and the ``min-max'' approach}

We now turn our attention to a different approach to the notion of robustness known as robust optimization \cite{RObook,ROsurvey}. In contrast with robust statistics, robust optimization typically replaces  distributional assumptions with a new primitive, namely, the deterministic notion of an \emph{uncertainty set}. Further, in robust optimization one considers a worst-case or pessimistic perspective and the focus is on perturbations from the nominal model (as opposed to possible gross corruptions as in robust statistics).

To be precise, one possible robust optimization model for linear regression takes form \cite{xu,RObook,bcrobreg}
\begin{equation}\label{eqn:roprimitive}
\min_\bb\max_{\D\in\U} \frac{1}{2}\|\y-(\X+\D)\bb\|_2^2,
\end{equation}
where $\U\sub\mathbb{R}^{n\times p}$ is a (deterministic) uncertainty set that captures the possible deviations of the model (from the nominal data $\X$). 
Note the immediate contrast with the robust models considered earlier (LTS and total least squares in \eqref{eqn:introlts} and \eqref{eqn:introeiv}, respectively) that take the min-min form; instead, robust optimization focuses on ``min-max'' robustness. For a related discussion contrasting the min-min approach with min-max, see \cite{worstbest,optimisticrobust} and references therein.

One of the attractive features of the min-max formulation is that it gives a re-interpretation of  several statistical regularization methods. For example, the usual Lasso (problem \eqref{eqn:a1} with $R=\mu\ell_1$) can be expressed in the form \eqref{eqn:roprimitive} for a specific choice of uncertainty set:
\begin{proposition}[e.g. \cite{xu,RObook}]\label{prop:lasso}
Problem \eqref{eqn:roprimitive} with uncertainty set $\U = \{\D: \|\D_i\|_2\leq \mu\;\forall i\}$
is equivalent to the Lasso, i.e., problem \eqref{eqn:a1} with $R(\bb)=\mu\|\bb\|_1$, where $\D_i$ denotes the $i$th column of $\D$.
\end{proposition}
\noindent For further discussion of the robust optimization approach as applied to statistical problems, see \cite{bcrobreg} and references therein.

\subsubsection*{Other min-max models of robustness}

We close our discussion of robustness by considering another example of min-max robustness that is of particular relevance to the trimmed Lasso. In particular, we consider problem \eqref{eqn:a1} with the SLOPE (or OWL) penalty \cite{slope,owl}, namely,
$$R_{\textsc{SLOPE}(\mb w)}(\bb) = \sum_{i=1}^p w_i |\beta_{(i)}|, $$
where $\mb w$ is a (fixed) vector of weights with $w_1\geq w_2\geq \cdots\geq w_p\geq 0$ and $w_1>0$. In its simplest form, the SLOPE penalty has weight vector $\tilde{\w}$, where $\tilde{w}_1=\cdots=\tilde{w}_k=1$, $\tilde{w}_{k+1}=\cdots=\tilde{w}_p=0$, in which case we have the identity
$$R_{\textsc{SLOPE}(\tilde{\w})}(\bb)  =\| \bb \|_{1} - T_{k}({\bb}).$$

There are some apparent similarities but also subtle differences between the SLOPE penalty and the trimmed Lasso. From a high level, while the trimmed Lasso focuses on the smallest magnitude entries of $\bb$, the SLOPE penalty in its simplest form focuses on the \emph{largest} magnitude entries of $\bb$. As such, the trimmed Lasso is generally nonconvex, while the SLOPE penalty is always convex; consequently, the techniques for solving the related estimation problems will necessarily be different.

Finally, we note that the SLOPE penalty can be considered as a min-max model of robustness for a particular choice of uncertainty set:

\begin{proposition}\label{prop:slope}
Problem \eqref{eqn:roprimitive} with uncertainty set
$$\U =\left\{\D :
\begin{array}{c}
\D \text{ has at most $k$ nonzero}\\
\text{columns and } \|\D_i\|_2\leq \mu\;\forall i
\end{array}\right\}
$$
is equivalent to problem \eqref{eqn:a1} with $R(\bb)=\mu R_{\textsc{SLOPE}(\tilde{\w})}(\bb) $, where $\tilde{w}_1=\cdots=\tilde{w}_k=1$ and $\tilde{w}_{k+1}=\cdots=\tilde{w}_p=0$.
\end{proposition}
\noindent We return to this particular choice of uncertainty set later. (For completeness, we include a more general min-max representation of SLOPE in Appendix \ref{app:slope}.)

\subsection*{Computation and Algorithms}

Broadly speaking, there are numerous distinct approaches to algorithms for solving problems of the form \eqref{eqn:a10}--\eqref{eqn:a1} for various choices of $R$. We do not attempt to provide a comprehensive list of such approaches for general $R$, but we will discuss existing approaches for the trimmed Lasso and closely related problems. Approaches typically take one of two forms: heuristic or exact.

\subsubsection*{Heuristic techniques}

Heuristic approaches to solving problems \eqref{eqn:a10}--\eqref{eqn:a1} often use techniques from convex optimization \cite{BV2004}, such as proximal gradient descent or coordinate descent (see \cite{scad,sparsenet}). Typically these techniques are coupled with an analysis of  local or global behavior of the algorithm. For example, global behavior is often considered under additional restrictive assumptions on the underlying data; unfortunately, verifying such assumptions can be as difficult as solving the original nonconvex problem. (For example, consider the analogy with compressed sensing \cite{crt,donoho1,gitta} and the hardness of verifying whether underlying assumptions hold \cite{tillman,bandeira}).

There is also extensive work studying the local behavior (e.g. stationarity) of heuristic approaches to these problems. For the specific problems \eqref{eqn:a10} and \eqref{eqn:a10p}, the behavior of augmented Lagrangian methods \cite{admmsilva,admmteng} and complementarity constraint techniques \cite{mpccportfolio,burdakov,compconl0,asc} have been considered. For other local approaches, see \cite{folded}.

\subsubsection*{Exact techniques}

One of the primary drawbacks of heuristic techniques is that it can often be difficult to verify the degree of suboptimality of the estimators obtained. For this reason, there has been an increasing interest in studying the behavior of exact algorithms for providing certifiably optimal solutions to problems of the form \eqref{eqn:a10}--\eqref{eqn:a1} \cite{bkm,bmlqs,mipgo,discretedantzig}. Often these approaches make use of techniques from \emph{mixed integer optimization} (``MIO'')\cite{bonami} which are implemented in a variety of software, e.g. Gurobi \cite{gurobi}. The tradeoff with such approaches is that they typically carry a heavier computational burden than convex approaches. For a discussion of the application of MIO  in statistics, see \cite{bkm,bmlqs,mipgo,discretedantzig}.

\subsection*{What this paper is about}

In this paper, we focus on a detailed analysis of the trimmed Lasso, especially with regard to its properties and its relation to existing methods. In particular, we explore the trimmed Lasso from two perspectives: that of sparsity as well as that of robustness. We summarize our contributions as follows:

\begin{enumerate}

\item We study the robustness of the trimmed Lasso penalty. In particular, we provide several min-min robustness representations of it. We first show that the same choice of uncertainty set that leads to the SLOPE penalty in the min-max robust model \eqref{eqn:roprimitive} gives rise to the trimmed Lasso in the corresponding min-min robust problem \eqref{eqn:introeivcon} (with an additional regularization term). This gives an interpretation of the SLOPE and trimmed Lasso as a complementary pair of penalties, one under a pessimistic (min-max) model and the other under an optimistic (min-min) model.

Moreover, we show another min-min robustness interpretation of the trimmed Lasso by comparison with the ordinary Lasso. In doing so, we further highlight the nature of the trimmed Lasso and its relation to the LTS problem \eqref{eqn:introlts}.

\item We provide a detailed analysis on the connection between estimation approaches using the trimmed Lasso and separable penalty functions. In doing so, we show directly how penalties such as the trimmed Lasso can be viewed as a generalization of such existing approaches in certain cases. In particular, a trimmed-Lasso-like approach always subsumes its separable analogue, and the containment is strict in general. We also focus on the specific case of the clipped (or capped) Lasso \cite{cl}; for this we precisely characterize the relationship and provide a necessary and sufficient condition for the two approaches to be equivalent. In doing so, we highlight some of the limitations of an approach using a separable penalty function.

\item Finally, we describe a variety of algorithms, both existing and new, for trimmed Lasso estimation problems. We contrast two heuristic approaches for finding locally optimal solutions with exact techniques from mixed integer optimization that can be used to produce certificates of optimality for solutions found via the convex approaches. We also show that the convex envelope \cite{rockafeller} of the trimmed Lasso takes the form
$$\left(\|\bb\|_1 - k\right)_+,$$
where $(a)_+:=\max\{0,a\}$, a ``soft-thresholded'' variant of the ordinary Lasso. Throughout this section, we emphasize how techniques from convex optimization can be used to find high-quality solutions to the trimmed Lasso estimation problem. An implementation of the various algorithms presented herein can be found at
\begin{center}
\url{https://github.com/copenhaver/trimmedlasso}.
\end{center}

\end{enumerate}

\subsubsection*{Paper structure}

The structure of the paper is as follows. In Section \ref{sec:basic}, we study several properties of the trimmed Lasso, provide a few distinct interpretations, and highlight possible generalizations. In Section \ref{sec:rob}, we explore the trimmed Lasso in the context of robustness. Then, in Section \ref{sec:ncpm}, we study the relationship between the trimmed Lasso and other nonconvex penalties. In Section \ref{sec:algs}, we study the algorithmic implications of the trimmed Lasso. Finally, in Section \ref{sec:conc} we share our concluding thoughts and highlight future directions.

%%%%%%%%%%%%%%%%%%%%%%%
%%% Trimmed Lasso
%%%%%%%%%%%%%%%%%%%%%%%

\section{Structural properties and interpretations}\label{sec:basic}

In this section, we provide further background on the trimmed Lasso: its motivations, interpretations, and generalizations. Our remarks in this section are broadly grouped as follows: in Section \ref{ssec:defn} we summarize the trimmed Lasso's basic properties as detailed in \cite{thiao,hempel,gotoh1,gotoh2}; we then turn our attention to an interpretation of the trimmed Lasso as a relaxation of complementarity constraints problems from optimization (Section \ref{ssec:compcon}) and as a variable decomposition method (Section \ref{ssec:vardecomp}); finally, in Sections \ref{ssec:gens} and \ref{ssec:dantzig} we highlight the key structural features of the trimmed Lasso by identifying possible generalizations of its definition and its application. These results augment the existing literature by giving a deeper understanding of the trimmed Lasso and provide a basis for further results in Sections \ref{sec:rob} and \ref{sec:ncpm}.

\subsection{Basic observations}\label{ssec:defn}

We begin with a summary of some of the basic properties of the trimmed Lasso as studied in \cite{thiao,hempel,gotoh1}.
First of all, let us also include another representation of $\tke$:
\begin{lemma}\label{lemma:miprep}
For any $\bb$,
$$\begin{array}{lll}
\tk{\bb} =  \smash{\ds\min_{\substack{I\sub\{1,\ldots,p\}:\\|I| = p-k}} \sum_{i\in I} |\beta_i| } = &\ds\min_{\zz} & \langle\zz,|\bb|\rangle\\
 & \st& \ds\sum_{i} z_i =p- k\\ & &\ds \zz\in\{0,1\}^p,
\end{array}$$
where $|\bb|$ denotes the vector whose entries are the absolute values of the entries of $\bb$.
\end{lemma}
\noindent In other words, the trimmed Lasso can be represented using auxiliary binary variables.

Now let us consider the problem
\begin{equation*}
%\tla{\lambda,k}\quad\quad
\min_{\bb} \frac{1}{2} \|\y-\X\bb\|_2^2+\lambda \tk{\bb},\tag{$\textsc{TL}_{\lambda,k}$}
\end{equation*}
where $\lambda>0$ and $k\in\{0,1,\ldots,p\}$ are parameters. Based on the definition of $\tke$, we have the following:
\begin{lemma}\label{lemma:vdrep}
The problem $\tla{\lambda,k}$ can be rewritten exactly in several equivalent forms:
\begin{align*}
\tla{\lambda,k} &=\min_{\substack{\bb,\ph:\\\|\ph\|_0\leq k}}\frac{1}{2}\|\y-\X\bb\|^2 + \lambda \|\bb-\ph\|_1\nonumber\\
&= \min_{\substack{\bb,\ph,\be:\\\bb=\ph+\be\\\|\ph\|_0\leq k}}\frac{1}{2}\|\y-\X\bb\|^2 + \lambda \|\be\|_1\nonumber\\
& =\min_{\substack{\ph,\be:\\\|\ph\|_0\leq k}}\frac{1}{2}\|\y-\X(\ph+\be)\|^2 + \lambda \|\be\|_1%\label{eqn:dtl}.
\end{align*}
\end{lemma}

\subsubsection*{Exact penalization}

Based on the definition of $T_k$, it follows that $T_k(\bb)=0$ if and only if $\|\bb\|_0\leq k$. Therefore, one can rewrite problem \eqref{eqn:a10} as
\begin{equation*}
\min_{T_k(\bb)=0} \frac{1}{2}\|\y-\X\bb\|_2^2.
\end{equation*}
In Lagrangian form, this would suggest an approximation for \eqref{eqn:a10} of the form
\begin{equation*}
\min_{\bb} \frac{1}{2} \|\y-\X\bb\|_2^2 + \lambda T_k(\bb),
\end{equation*}
where $\lambda>0$. As noted in the introduction, this approximation is in fact exact (in the sense of \cite{bert76,bertexact}), summarized in the following theorem; for completeness, we include in Appendix \ref{app:proof} a full proof that is distinct from that in \cite{gotoh1}.\footnote{The presence of the additional regularizer $\eta\|\bb\|_1$ can be interpreted in many ways. For our purposes, it serves to make the problems well-posed.}

\begin{theorem}[\emph{c.f.} \cite{gotoh1}]\label{thm:exactEquiv}
For any fixed $k\in\{0,1,2,\ldots,p\}$, $\eta>0$, and problem data $\y$ and $\X$, there exists some $\barl=\barl(\y,\X)>0$ so that for all $\lambda>\barl$, the problems 
$$\ds\min_{\bb} \frac{1}{2}\|\y-\X\bb\|_2^2 + \lambda \tk{\bb} + \eta \|\bb\|_1 $$
and
$$\begin{array}{ll}
\ds\min_{\bb}& \frac{1}{2}\|\y-\X\bb\|_2^2 + \eta\|\bb\|_1\\
\st& \|\bb\|_0\leq k
\end{array}$$
have the same optimal objective value and the same set of optimal solutions.
\end{theorem}

The direct implication is that trimmed Lasso leads to a continuum (over $\lambda$) of relaxations to the best subset selection problem starting from ordinary least squares estimation; further, best subset selection lies on this continuum for $\lambda$ sufficiently large.

\subsection{A complementary constraints viewpoint}\label{ssec:compcon}

We now turn our attention to a new perspective on the trimmed Lasso as considered via mathematical programming with complementarity constraints (``MPCCs'') \cite{scholtes,mpcclin,kanzow0,kanzow1,kanzow2,burdakov}, sometimes also referred to as mathematical programs with equilibrium constraints \cite{bilevel}. By studying this connection, we will show that a penalized form of a common relaxation scheme for MPCCs leads directly to the trimmed Lasso penalty. This gives a distinctly different optimization perspective on the trimmed Lasso penalty.

As detailed in \cite{mpccportfolio,burdakov,compconl0}, the problem \eqref{eqn:a10} can be exactly rewritten as
\begin{equation}\label{eqn:BSSr}
\begin{array}{ll}
\ds\min_{\bb,\zz}&\ds \frac{1}{2}\|\y-\X\bb\|_2^2 \\%+ R(\bb)\\
\st& \sum_iz_i=p-k\\
& \zz\in[0,1]^p\\
& z_i\beta_i = 0.
\end{array}
\end{equation}
by the inclusion of auxiliary variables $\zz\in[0,1]^p$. In essence, the auxiliary variables replace the combinatorial constraint $\|\bb\|_0\leq k$ with \emph{complementarity} constraints of the form $z_i\beta_i=0$. Of course, the problem as represented in \eqref{eqn:BSSr} is still not directly amenable to convex optimization techniques.

As such, relaxation schemes can be applied to \eqref{eqn:BSSr}. One popular  method from the MPCC literature is the Scholtes-type relaxation \cite{kanzow1}; applied to \eqref{eqn:BSSr} as in \cite{burdakov,compconl0}, this takes the form
\begin{equation}\label{eqn:BSSr2}
\begin{array}{ll}
\ds\min_{\bb,\zz}& \ds\frac{1}{2}\|\y-\X\bb\|_2^2\\% + R(\bb)\\
\st& \sum_iz_i=p-k\\
& \zz\in[0,1]^p\\
& |z_i\beta_i|\leq t,
\end{array}
\end{equation}
where $t>0$ is some fixed numerical parameter which controls the strength of the relaxation, with $t=0$ exactly recovering \eqref{eqn:BSSr}. In the traditional MPCC context, it is standard to study local optimality and stationarity behavior of solutions to \eqref{eqn:BSSr2} as they relate to the original problem \eqref{eqn:a10}, \emph{c.f.} \cite{compconl0}.

Instead, let us consider a different approach. In particular, consider a penalized, or Lagrangian, form of the Scholtes relaxation \eqref{eqn:BSSr2}, namely,
\begin{equation}\label{eqn:BSSr3}
\begin{array}{ll}
\ds\min_{\bb,\zz}& \ds\frac{1}{2}\|\y-\X\bb\|_2^2   + \lambda\sum_i (|z_i\beta_i|-t)\\% + R(\bb)\\
\st& \sum_iz_i=p-k\\
& \zz\in[0,1]^p
\end{array}
\end{equation}
for some fixed $\lambda\geq0$.\footnote{To be precise, this is a \emph{weaker} relaxation than if we had separate dual variables $\lambda_i$ for each constraint $|z_i\beta_i|\leq t$, at least in theory.} Observe that we can minimize \eqref{eqn:BSSr3} with respect to $\zz$ to obtain the equivalent problem
$$\min_{\bb} \ds \frac{1}{2}\|\y-\X\bb\|_2^2  + \lambda T_k(\bb)  - p\lambda t,$$
which is precisely problem $\tl$ (up to the fixed additive constant). In other words, the trimmed Lasso can also be viewed as arising directly from a  penalized form of the MPCC relaxation, with auxiliary variables eliminated. This gives another view on Lemma \ref{lemma:miprep} which gave a representation of $\tke$ using auxiliary binary variables.

\subsection{Variable decomposition}\label{ssec:vardecomp}

To better understand the relation of the trimmed Lasso to existing methods, it is also useful to consider alternative representations. Here we focus on representations which connect it to variable decomposition methods. Our discussion here is an extended form of related discussions in \cite{hempel,gotoh1,gotoh2}.

To begin, we return to the final representation of the trimmed Lasso problem as shown in Lemma \ref{lemma:vdrep}, viz.,
\begin{equation}\label{eqn:dtl}
\tla{\lambda,k}=\min_{\substack{\ph,\be:\\\|\ph\|_0\leq k}}\frac{1}{2}\|\y-\X(\ph+\be)\|^2 + \lambda \|\be\|_1.%\tag{\ref{eqn:dtl}}
\end{equation}
We will refer to $\tla{\lambda,k}$ in the form \eqref{eqn:dtl}
%\begin{equation}\label{eqn:dtl} % decomposed trimmed lasso
%\min_{\substack{\ph,\be:\\\|\ph\|_0\leq k}}\frac{1}{2}\|\y-\X(\ph+\be)\|^2 + \lambda \|\be\|_1
%\end{equation}
as the \emph{split} or \emph{decomposed} representation of the problem. This is because in this form it is clear that we can think about estimators $\bb$ found via $\tla{\lambda,k}$ as being decomposed into two different estimators: a sparse component $\ph$ and another component $\be$ with small $\ell_1$ norm (as controlled via $\lambda$).

Several remarks are in order. First, the decomposition of $\bb$ into $\bb=\ph+\be$ is truly a decomposition in that if $\bb^*$ is an optimal solution to $\tla{\lambda,k}$ with $(\ph^*,\be^*)$ a corresponding optimal solution to the split representation of the problem \eqref{eqn:dtl}, then one must have that $\phi_i^*\epsilon_i^*=0$ for all $i\in\{1,\ldots,p\}$. In other words, the supports of $\ph$ and $\be$ do not overlap; therefore, $\bb^*=\ph^*+\be^*$ is a genuine decomposition.

Secondly, the variable decomposition \eqref{eqn:dtl} suggests that the problem of finding the $k$ largest entries of $\bb$ (i.e., finding $\ph$) can be solved as a best subset selection problem with a (possibly different) convex loss function (without $\be$). To see this, observe that the problem of finding $\ph$ in \eqref{eqn:dtl} can be written as the problem
$$\min_{\|\ph\|_0\leq k} \widetilde{L}(\ph),$$
where
$$\widetilde{L}(\ph) = \min_{\be} \frac{1}{2}\|\y-\X(\ph+\be)\|_2^2 + \lambda \|\be\|_1.$$
Using theory on duality for the Lasso problem \cite{lassodual}, one can argue that $\widetilde{L}$ is itself a convex loss function. Hence, the variable decomposition gives some insight into how the largest $k$ loadings for the trimmed Lasso relates to solving a related sparse estimation problem.

\subsubsection*{A view towards matrix estimation}

Finally, we contend that the variable decomposition of $\bb$ as a sparse component $\ph$ plus a ``noise'' component $\be$ with small norm is a natural and useful analogue of corresponding decompositions in the matrix estimation literature, such as in factor analysis \cite{mardia,anderson2003,barth} and robust Principal Component Analysis \cite{candesrpca}. For the purposes of this paper, we will focus on the analogy with factor analysis.

Factor analysis is a classical multivariate statistical method for decomposing the covariance structure of random variables; see \cite{fabcm} for an overview of modern approaches to factor analysis. Given a covariance matrix $\s\in\R^{p\times p}$, one is interested in describing it as the sum of two distinct components: a low-rank component $\T$ (corresponding to a low-dimensional covariance structure common across the variables) and a diagonal component $\Ph$ (corresponding to individual variances unique to each variable)---in symbols, $\s=\T+\Ph$.

In reality, this \emph{noiseless} decomposition is often too restrictive (see e.g.\cite{guttman1958extent,shapirorankred,ten1998some}), and therefore it is often better to focus on finding a decomposition $\s=\T+\Ph+\N$, where $\N$ is a noise component with small norm. As in \cite{fabcm}, a corresponding estimation procedure can take the form
\begin{equation}\label{eqn:fa}
\begin{array}{ll}
\ds\min_{\T,\Ph}& \|\s-(\T+\Ph)\|\\
\st& \rank(\T)\leq k\\
& \Ph = \diag(\Phi_{11},\ldots,\Phi_{pp}) \psd\\
& \T\psd,
\end{array}
\end{equation}
where the constraint $\mb A\psd$ denotes that $\mb A$ is symmetric, positive semidefinite, and $\|\cdot\|$ is some norm. One of the attractive features of the estimation procedure \eqref{eqn:fa} is that for common choices of $\|\cdot\|$, it is possible to completely eliminate the combinatorial rank constraint and the variable $\T$ to yield a smooth (nonconvex) optimization problem with compact, convex constraints (see \cite{fabcm} for details).

This exact same argument can be used to motivate the appearance of the trimmed Lasso penalty. Indeed, instead of considering estimators $\bb$ which are exactly $k$-sparse (i.e., $\|\bb\|_0\leq k$), we instead consider estimators which are approximately $k$-sparse, i.e., $\bb=\ph+\be$, where $\|\ph\|_0\leq k$ and $\be$ has small norm. Given fixed $\bb$, such a procedure is precisely
$$\min_{\|\ph\|_0\leq k} \|\bb-\ph\|.$$
Just as the rank constraint is eliminated from \eqref{eqn:fa}, the sparsity constraint can be eliminated from this to yield a continuous penalty which precisely captures the quality of the approximation $\bb\approx\ph$. The trimmed Lasso uses the choice $\|\cdot\|=\ell_1$, although other choices are possible; see Section \ref{ssec:gens}.

This analogy with factor analysis is also useful in highlighting additional benefits of the trimmed Lasso. One of particular note is that it enables the direct application of existing convex optimization techniques to find high-quality solutions to $\tla{\lambda,k}$.

\subsection{Generalizations}\label{ssec:gens}

We close this section by considering some generalizations of the trimmed Lasso. These are particularly useful for connecting the trimmed Lasso to other penalties, as we will see later in Section \ref{sec:ncpm}.

As noted earlier, the trimmed Lasso measures the distance (in $\ell_1$ norm) from the set of $k$-sparse vectors; therefore, it is natural to inquire what properties other measures of distance might carry. In light of this, we begin with a definition:
\begin{definition}
Let $k\in\{0,1,\ldots,p\}$ and $g:\R_+\to\R_+$ be any unbounded, continuous, and strictly increasing function with $g(0)=0$. Define the corresponding $k$th projected penalty function, denoted $\pi_k^g$, as
$$\pi_k^g(\bb) = \min_{\|\ph\|_0\leq k}  \sum_i g(|\phi_i-\beta_i|).$$
\end{definition}

\noindent It is not difficult to argue that $\pi_k^g$ has as an equivalent definition
$$\pi_k^g(\bb) = \sum_{i>k} g(|\beta_{(i)}|).$$
As an example, $\pi_k^g$ is the trimmed Lasso penalty when $g$ is the absolute value, viz. $g(x)=|x|$, and so it is a special case of the projected penalties. Alternatively, suppose $g(x) = x^2/2$. In this case, we get a trimmed version of the ridge regression penalty: $\sum_{i>k} |\beta_{(i)}|^2/2$.

This class of penalty functions has one notable feature, summarized in the following result:\footnote{An extended statement of the convergence claim is included in Appendix \ref{app:proof}.}

\begin{proposition}\label{prop:asymp}
If $g:\R_+\to\R_+$ is an unbounded, continuous, and strictly increasing function with $g(0)=0$, then for any $\bb$, 
$\pi_k^g(\bb) = 0$ if and only if $\|\bb\|_0\leq k$. 
Hence, the problem $\ds\min_{\bb} \frac{1}{2}\|\y-\X\bb\|_2^2 + \lambda\pi_k^g(\bb)$ converges in objective value to $\ds\min_{\|\bb\|_0\leq k} \frac{1}{2}\|\y-\X\bb\|_2^2$ as $\lambda\to\infty$.
\end{proposition}

Therefore, any projected penalty $\pi_k^g$ results in the best subset selection problem \eqref{eqn:a10} asymptotically. While the choice of $g$ as the absolute value gives the trimmed Lasso penalty and leads to exact sparsity in the non-asymptotic regime (\emph{c.f.} Theorem \ref{thm:exactEquiv}) , Proposition \ref{prop:asymp} suggests that the projected penalty functions have potential utility in attaining approximately sparse estimators. We will return to the penalties $\pi_k^g$ again in Section \ref{sec:ncpm} to connect the trimmed Lasso to nonconvex penalty methods.

Before concluding this section, we briefly consider a projected penalty function that is different than the trimmed Lasso. As noted above, if $g(x) = x^2/2$, then the corresponding penalty function is the trimmed ridge penalty $\sum_{i>k} |\beta_{(i)}|^2/2$.
The estimation procedure is then
$$\min_{\bb} \frac{1}{2} \|\y-\X\bb\|_2^2 +\frac{\lambda}{2} \sum_{i>k}|\beta_{(i)}|^2,$$
or equivalently in decomposed form (\emph{c.f.} Section \ref{ssec:vardecomp}),\footnote{Interestingly, if one considers this trimmed ridge regression problem and uses convex envelope techniques \cite{rockafeller,BV2004} to relax the constraint $\|\ph\|_0\leq k$, the resulting problem takes the form $\min_{{ \ph ,\be }} \|\y-\X(\ph+\be)\|_2^2/2 + \lambda \|\be\|_2^2 + \tau\|\ph\|_1$, a sort of ``split'' variant of the usual elastic net \cite{zhel}, another popular convex method for sparse modeling.}
$$\min_{\substack{\ph,\be:\\\|\ph\|_0\leq k}}\frac{1}{2}\|\y-\X(\ph+\be)\|_2^2 +\frac{ \lambda}{2} \|\be\|_2^2.$$
It is not difficult to see that the variable $\be$ can be eliminated to yield
\begin{equation}\label{eqn:ridge}
\min_{\|\ph\|_0\leq k} \frac{1}{2}\left\|\A(\y-\X\ph)\right\|_2^2,
\end{equation}
where $\A = (\I-\X(\X'\X+\lambda\I)^{-1}\X')^{1/2}$. It follows that the largest $k$ loadings are found via a modified best subset selection problem under a different loss function---precisely a variant of the $\ell_2$ norm. This is in the same spirit of observations made in Section \ref{ssec:vardecomp}.

\begin{obs}\label{obs:ridge}
An obvious question is whether the norm in \eqref{eqn:ridge} is genuinely different. Observe that this loss function is the same as the usual $\ell_2^2$ loss if and only if $\A'\A$ is a non-negative multiple of the identity matrix. It is not difficult to see that this is true iff $\X'\X$ is a non-negative multiple of the identity. In other words, the loss function in \eqref{eqn:ridge} is the same as the usual ridge regression loss if and only if $\X$ is (a scalar multiple of) an orthogonal design matrix.
\end{obs}

\subsection{Other applications of the trimmed Lasso: the (Discrete) Dantzig Selector}\label{ssec:dantzig}

The above discussion which pertains to the least squares loss data-fidelity term can be generalized to other loss functions as well.
For example, let us consider a data-fidelity term given by the 
maximal absolute inner product between the features and residuals, given by $\|\X'(\y-\X\bb)\|_\infty$. An $\ell_{1}$-penalized version of this data-fidelity term, 
popularly known as the Dantzig Selector~\cite{dantzig2,dasso}, is given by the following linear optimization problem:
\begin{equation}\label{eqn-DS}
\min_\bb  \|\X'(\y-\X\bb)\|_\infty  + \mu\|\bb\|_1.
\end{equation}
Estimators found via \eqref{eqn-DS} have statistical properties similar to the Lasso.
%%%The above estimator can be obtained via linear optimization. 
Further, problem \eqref{eqn-DS} may be interpreted as an $\ell_{1}$-approximation to the cardinality constrained version: 
\begin{equation}\label{eqn-DS-L0}
\min_{\|\bb\|_0\leq k}  \|\X'(\y-\X\bb)\|_\infty,
\end{equation}
that is, the Discrete Dantzig Selector, recently proposed and studied in~\cite{discretedantzig}. The statistical properties of~\eqref{eqn-DS-L0} are similar to the best-subset selection problem \eqref{eqn:a10}, but  may be more attractive from a computational viewpoint as it relies on mixed integer \emph{linear} optimization as opposed to mixed integer \emph{conic} optimization (see \cite{discretedantzig}).

The trimmed Lasso penalty can also be applied to the data-fidelity term $\|\X'(\y-\X\bb)\|_\infty$, leading to the following estimator:
$$\min_\bb  \|\X'(\y-\X\bb)\|_\infty+\lambda\tk{\bb}+\mu\|\bb\|_1.$$
Similar to the case of the least squares loss function, the above estimator yields $k$-sparse solutions for any $\mu>0$ and for $\lambda>0$ sufficiently large.\footnote{For the same reason, but instead with the usual Lasso objective, the proof of Theorem \ref{thm:exactEquiv} (see Appendix \ref{app:proof}) could be entirely omitted; yet, it is instructive to see in the proof there that the trimmed Lasso truly does set the \emph{smallest} entries to zero, and not simply all entries (when $\lambda$ is large) like the Lasso.} While this claim follows \emph{a fortiori} by appealing to properties of the Dantzig selector, it nevertheless highlights how any exact penalty method with a separable penalty function can be turned into a trimmed-style problem which offers direct control over the sparsity level.

%%%%%%%%%%%%%%%%%%%%%%%
%%% Connection to robustness
%%%%%%%%%%%%%%%%%%%%%%%

\section{A perspective on robustness}\label{sec:rob}

We now turn our attention to a deeper exploration of the robustness properties of the trimmed Lasso. We begin by studying the min-min robust analogue of the min-max robust SLOPE penalty; in doing so, we show under which circumstances this analogue is the trimmed Lasso problem. Indeed, in such a regime, the trimmed Lasso can be viewed as an optimistic counterpart to the robust optimization view of the SLOPE penalty. Finally, we turn our attention to an additional min-min robust interpretation of the trimmed Lasso in direct correspondence with the least trimmed squares estimator shown in \eqref{eqn:introlts}, using the ordinary Lasso as our starting point.

\subsection{The trimmed Lasso as a min-min robust analogue of SLOPE}

We begin by reconsidering the uncertainty set that gave rise to the SLOPE penalty via the min-max view of robustness as considered in robust optimization:
$$\U_k^\lambda:=\left\{\D :
\begin{array}{c}
\D \text{ has at most $k$ nonzero}\\
\text{columns and } \|\D_i\|_2\leq \lambda\;\forall i
\end{array}\right\}.
$$
As per Proposition \ref{prop:slope}, the min-max problem \eqref{eqn:roprimitive}, viz., 
\begin{equation*}%\label{eqn:roprimitive}
\min_\bb\max_{\D\in\U_k^\lambda} \frac{1}{2}\|\y-(\X+\D)\bb\|_2^2%\tag{\ref{eqn:roprimitive}}
\end{equation*}
is equivalent to the SLOPE-penalized problem
\begin{equation}\label{eqn:slopepen}
\min_\bb \frac{1}{2}\|\y-\X\bb\|_2^2 + \lambda R_{\textsc{SLOPE}(\tilde{\w})} (\bb).
\end{equation}
for the specific choice of $\tilde{\w}$ with $\tilde w_1=\cdots=\tilde w_k=1$ and $\tilde w_{k+1}=\cdots=\tilde w_{p}=0$.

Let us now consider the form of the min-min robust analogue of the the problem \eqref{eqn:roprimitive} for this specific choice of uncertainty set. As per the discussion in Section \ref{sec:intro}, the min-min analogue takes the form of problem \eqref{eqn:introeivcon}, i.e., a variant of total least squares:
\begin{equation*}%\label{eqn:eivcon}
\min_\bb \min_{\D\in\U_k^\lambda} \frac{1}{2}\|\y-(\X+\D)\bb\|_2^2,
\end{equation*}
or equivalently as the linearly homogenous problem\footnote{In what follows, the linear homogeneity is useful primarily for simplicity of analysis, \emph{c.f.} \cite[ch. 12]{RObook}. Indeed, the conversion to linear homogeneous functions is often hidden in equivalence results like Proposition \ref{prop:slope}.}
\begin{equation}\label{eqn:eivconhom}
\min_\bb \min_{\D\in\U_k^\lambda} \|\y-(\X+\D)\bb\|_2.
\end{equation}

\noindent It is useful to consider problem \eqref{eqn:eivconhom} with an explicit penalization (or regularization) on $\bb$:%. For example, consider the $\ell_1$-penalized form
\begin{equation}\label{eqn:eivconhompen}
\min_\bb \min_{\D\in\U_k^\lambda} \|\y-(\X+\D)\bb\|_2 + r(\bb),
\end{equation}
where $r(\cdot)$ is, say, a norm (the use of lowercase is to distinguish from the function $R$ in Section \ref{sec:intro}).

As described in the following theorem, this min-min robustness problem \eqref{eqn:eivconhompen} is equivalent to the trimmed Lasso problem for specific choices of $r$. The proof is contained in Appendix \ref{app:proof}.

\begin{theorem}\label{thm:robeivInterp}

For any $k$, $\lambda>0$, and norm $r$, the problem \eqref{eqn:eivconhompen} can be rewritten exactly as
%\begin{equation}\label{eqn:pfsupp3}
%\min_{\bb} \left(\|\y-\X\bb\|_2 - \lambda \sum_{i=1}^k |\beta_{(i)}|   \right)_+ + \tau\|\bb\|_1,
%\end{equation}
%where $(a)_+:=\max\{0,a\}$.
\begin{equation*}%\label{eqn:a3}
\begin{array}{ll}
\ds\min_{\bb } &\ds\|\y-\X\bb\|_2 + r(\bb) - \lambda \sum_{i=1}^k|\beta_{(i)}| \\
\st & \ds\lambda\sum_{i=1}^k|\beta_{(i)}| \leq \|\y-\X\bb\|_2.
\end{array}
\end{equation*}
\end{theorem}

We have the following as an immediate corollary:

\begin{corollary}\label{cor:slope}
For the choice of $r(\bb) = \tau \|\bb\|_1$, where $\tau > \lambda$, the problem \eqref{eqn:eivconhompen} is precisely
\begin{equation}\label{eqn:cor1}
\begin{array}{ll}
\ds\min_{\bb } &\ds\|\y-\X\bb\|_2 + (\tau-\lambda)\|\bb\|_1+ \lambda \tk{\bb} \\
\st & \ds\lambda\sum_{i=1}^k|\beta_{(i)}| \leq \|\y-\X\bb\|_2.
\end{array}
\end{equation}
In particular, when $\lambda>0$ is small, it is approximately equal (in a precise sense)\footnote{For a precise characterization and extended discussion, see Appendix \ref{app:proof} and Theorem \ref{thm:corprecise}. The informal statement here is sufficient for the purposes of our present discussion.} to the trimmed Lasso problem
$$\min_{\bb } \ds\|\y-\X\bb\|_2 + (\tau-\lambda)\|\bb\|_1+ \lambda \tk{\bb}.$$
\end{corollary}

In words, the min-min problem \eqref{eqn:eivconhompen} (with an $\ell_1$ regularization on $\bb$) can be written as a variant of a trimmed Lasso problem, subject to an additional constraint. It is instructive to consider both the objective and the constraint of problem \eqref{eqn:cor1}. To begin, the objective has a combined penalty on $\bb$ of $(\tau-\lambda)\|\bb\|_1 + \lambda \tk{\bb}$. This can be thought of as the more general form of the penalty $\tke$. Namely, one can consider the penalty $T_{\x}$ (with $0\leq x_1\leq x_2\leq\cdots\leq x_p$ fixed) defined as
$$T_{\x} (\bb) := \sum_{i=1}^p x_i |\beta_{(i)}|.$$
In this notation, the objective of \eqref{eqn:cor1} can be rewritten as $\|\y-\X\bb\|_2 + T_{\x}(\bb)$, with
$$\x=(\underbrace{\tau-\lambda,\ldots,\tau-\lambda}_{k \text{ times}},\underbrace{\tau,\ldots,\tau}_{p-k \text{ times}}).$$
In terms of the constraint of problem \eqref{eqn:cor1}, note that it takes the form of a model-fitting constraint: namely, $\lambda$ controls a trade-off between model fit $\|\y-\X\bb\|_2$ and model complexity measured via the SLOPE norm $\sum_{i=1}^k |\beta_{(i)}|$.

Having described the structure of problem \eqref{eqn:cor1}, a few remarks are in order. First of all, the trimmed Lasso problem (with an additional $\ell_1$ penalty on $\bb$) can be interpreted as (a close approximation to) a min-min robust problem, at least in the regime when $\lambda$ is small; this provides an interesting contrast to the sparse-modeling regime when $\lambda$ is large (\emph{c.f.} Theorem \ref{thm:exactEquiv}). Moreover, the trimmed Lasso is a min-min robust problem in a way that is the \emph{optimistic} analogue of its min-max counterpart, namely, the SLOPE-penalized problem \eqref{eqn:slopepen}. Finally, Theorem \ref{thm:robeivInterp} gives a natural representation of the trimmed Lasso problem in a way that directly suggests why methods from difference-of-convex optimization \cite{dcSummary} are relevant (see Section \ref{sec:algs}).

\subsubsection*{The general SLOPE penalty}

Let us briefly remark upon SLOPE in its most general form (with general $\mb w$); again we will see that this leads to a more general trimmed Lasso as its (approximate) min-min counterpart. In its most general form, the SLOPE-penalized problem \eqref{eqn:slopepen} can be written as the min-max robust problem \eqref{eqn:roprimitive} with choice of uncertainty set
$$\U_\w^\lambda =\left\{\D : \|\D\ph\|_2 \leq  \lambda \sum_i w_i|\phi_{(i)}|\;\forall\ph\right\}
$$
(see Appendix \ref{app:slope}). In this case, the penalized, homogenized min-min robust counterpart, analogous to problem \eqref{eqn:eivconhompen}, can be written as follows:

\begin{proposition}\label{prop:robeivslope}
For any $k$, $\lambda>0$, and norm $r$, the problem
\begin{equation}\label{eqn:auxslopepf1}
\min_\bb \min_{\D\in\U_\w^\lambda} \|\y-(\X+\D)\bb\|_2 + r(\bb)
\end{equation}
can be rewritten exactly as
\begin{equation*}%\label{eqn:a3}
\begin{array}{ll}
\ds\min_{\bb } &\ds\|\y-\X\bb\|_2 + r(\bb) - \lambda R_{\textsc{SLOPE}(\w)}(\bb) \\
\st & \ds\lambda R_{\textsc{SLOPE}(\w)}(\bb) \leq \|\y-\X\bb\|_2.
\end{array}
\end{equation*}
For the choice of $r(\bb) = \tau \|\bb\|_1$, where $\tau > \lambda w_1$, the problem \eqref{eqn:auxslopepf1} is
\begin{equation*}%\label{eqn:a3}
\begin{array}{ll}
\ds\min_{\bb } &\ds\|\y-\X\bb\|_2 + T_{\tau\mb 1 - \lambda \w}(\bb) \\
\st & \ds\lambda R_{\textsc{SLOPE}(\w)}(\bb) \leq \|\y-\X\bb\|_2.
\end{array}
\end{equation*}
%where $T_{\x}$ (with $0\leq x_1\leq x_2\leq\cdots\leq x_p$) is a generalized trimmed Lasso penalty defined as
%$$T_{\x} (\bb) := \sum_{i=1}^p x_i |\beta_{(i)}|.$$
In particular, when $\lambda>0$ is sufficiently small, problem \eqref{eqn:auxslopepf1} is approximately equal to the generalized trimmed Lasso problem
$$\min_{\bb } \ds\|\y-\X\bb\|_2 + T_{\tau\mb 1 - \lambda \w}(\bb).$$
\end{proposition}

Put plainly, the general form of the SLOPE penalty leads to a generalized form of the trimmed Lasso, precisely as was true for the simplified version considered in Theorem \ref{thm:robeivInterp}.

\subsection{Another min-min interpretation}

We close our discussion of robustness by considering another min-min representation of the trimmed Lasso. We use the ordinary Lasso problem as our starting point and show how a modification in the same spirit as the min-min robust least trimmed squares estimator in \eqref{eqn:introlts} leads directly to the trimmed Lasso.

To proceed, we begin with the usual Lasso problem
\begin{equation}\label{eqn:lasso}
\min_\bb \frac{1}{2}\|\y-\X\bb\|_2^2 + \lambda\|\bb\|_1.
\end{equation}
As per Proposition \ref{prop:lasso}, this problem is equivalent to the min-max robust problem \eqref{eqn:roprimitive} with uncertainty set
$\U = \cL^\lambda = \{\D:\|\D_i\|_2 \leq \lambda\;\forall i\}$:
\begin{equation}\label{eqn:lassorob}
\min_\bb\max_{\D\in\cL^\lambda} \frac{1}{2}\|\y-(\X+\D)\bb\|_2^2.
\end{equation}
In this view, the usual Lasso \eqref{eqn:lasso} can be thought of as a least squares method which takes into account certain feature-wise adversarial perturbations of the matrix $\X$. The net result is that the adversarial approach penalizes all loadings equally (with coefficient $\lambda$).

Using this setup and Theorem \ref{thm:exactEquiv},
we can re-express the trimmed Lasso problem $\tla{\lambda,k}$ in the equivalent min-min form
\begin{equation}\label{eqn:a2}
\min_\bb\min_{\substack{\\I\sub\{1,\ldots,p\}:\\|I|=p-k}} \max_{\D\in\cL^\lambda_I} \frac{1}{2}\|\y-(\X+\D)\bb\|_2^2,
\end{equation}
where $\cL^\lambda_I\sub \cL^\lambda$ requires that the columns of $\D\in\cL^\lambda_I$ are supported on $I$:
$$\cL^\lambda_I = \{\D: \|\D_i\|_2\leq \lambda \;\forall i, \; \D_i = \mb 0\;\forall i\notin I\}.$$
While the adversarial min-max approach in problem \eqref{eqn:lassorob} would attempt to ``corrupt'' all $p$ columns of $\X$, in estimating $\bb$ we have the power to optimally discard $k$ out of the $p$ corruptions to the columns (corresponding to $I^c$). In this sense, the trimmed Lasso in the min-min robust form \eqref{eqn:a2} acts in a similar spirit to the min-min, robust-statistical least trimmed squares estimator shown in problem \eqref{eqn:introltsalt}.

%%%%%%%%%%%%%%%%%%%%%%%
%%% Connection to nonconvex penalty methods
%%%%%%%%%%%%%%%%%%%%%%%

\section{Connection to nonconvex penalty methods}\label{sec:ncpm}

In this section, we explore the connection between the trimmed Lasso and existing, popular nonconvex (component-wise separable) penalty functions used for sparse modeling.  We begin in Section \ref{ssec:ncpoverview} with a brief overview of existing approaches. In Section \ref{ssec:ncpreform} we then highlight how these relate to the trimmed Lasso, making the connection more concrete with examples in Section \ref{ssec:ncpeg}. Then in Section \ref{ssec:ncpgenerality} we exactly characterize the connection between the trimmed Lasso and the clipped Lasso \cite{cl}. In doing so, we show that the trimmed Lasso subsumes the clipped Lasso; further, we provide a necessary and sufficient condition for when the containment is strict. Finally, in Section \ref{ssec:ncpunbounded} we comment on the special case of unbounded penalty functions.

\subsection{Setup and Overview}\label{ssec:ncpoverview}

Our focus throughout will be the penalized $M$-estimation problem of the form
\begin{equation}\label{eqn:ncpm}%nonconvex penalty M estimation
\min_{\bb} L(\bb) + \sum_{i=1}^p \rho(|\beta_i|;\mu,\gamma),
\end{equation}
where $\mu$ represents a (continuous) parameter controlling the desired level of sparsity of $\bb$ and $\gamma$ is a parameter controlling the quality of the approximation of the indicator function $I\{|\beta|>0\}$. A variety of nonconvex penalty functions and their description in this format is shown in Table \ref{tab:ncp} (for a general discussion, see \cite{zhangzhang}). In particular, for each of these functions we observe that 
$$\lim_{\gamma\to\infty} \rho(|\beta|;\mu,\gamma) = \mu \cdot I\{|\beta|>0\}.$$
It is particularly important to note the \emph{separable} nature of the penalty functions appearing in \eqref{eqn:ncpm}---namely, each coordinate $\beta_i$ is penalized (via $\rho$) independently of the other coordinates.

Our primary focus will be on the bounded penalty functions (clipped Lasso, MCP, and SCAD), all of which take the form
\begin{equation}\label{eqn:pff}
\rho(|\beta|;\mu,\gamma) = \mu \min\{g(|\beta|;\mu,\gamma),1\}
\end{equation}
where $g$ is an increasing function of $|\beta|$. We will show that in this case, the problem \eqref{eqn:ncpm} can be rewritten exactly as an estimation problem with a (non-separable) trimmed penalty function:
\begin{equation}\label{eqn:ncpt}%nonconvex penalty trimmed estimation
\min_{\bb} L(\bb) + \mu\sum_{i=\ell+1}^p g(|\beta_{(i)} |)
\end{equation}
for some $\ell\in\{0,1,\ldots,p\}$ (note the appearance of the projected penalties $\pi_k^g$ as considered in Section \ref{ssec:gens}). In the process of doing so, we will also show that, in general, \eqref{eqn:ncpt} cannot be solved via the separable-penalty estimation approach of \eqref{eqn:ncpm}, and so the trimmed estimation problem leads to a richer class of models. Throughout we will often refer to \eqref{eqn:ncpt} (taken generically over all choices of $\ell$) as the \emph{trimmed counterpart} of the separable estimation problem \eqref{eqn:ncpm}.

\begin{table*}
\centering
  \begin{tabular}{| c | c | c | c|}
\hline 
Name  & Definition & Auxiliary Functions\\\hline\hline
Clipped Lasso &\multirow{2}{*}{$ \mu\min\{\gamma|\beta|,1\} $ } &\multirow{4}{*}{ \small $g_1(|\beta|) =  \left\{\begin{array}{rc} 2\gamma|\beta|-\gamma^2\beta^2, & |\beta|\leq 1/\gamma,\\1 ,& |\beta|>1/\gamma.\end{array}\right.$  }\\
\cite{cl} &  & \\\cline{1-2}
MCP & \multirow{2}{*}{$\mu\min\{g_1(|\beta|),1\}$ } &   \\
\cite{mcp} & & \\\cline{1-2}
SCAD & \multirow{2}{*}{$\mu\min\{g_2(|\beta|),1\}$} &  \multirow{6}{*}{ \small $ g_2(|\beta|) = \left\{
\begin{array}{rc}
|\beta|/(\gamma\mu),&  |\beta| \leq 1/\gamma,\\
 \frac{\beta^2 +(2/\gamma-4\mu\gamma)|\beta| +1/\gamma^2}{ 4\mu - 4\mu^2\gamma^2} , &  1/\gamma < |\beta| \leq 2\mu\gamma-1/\gamma,\\1, &  |\beta| > 2\mu\gamma-1/\gamma.
\end{array}\right.   $  }   \\
\cite{scad} & & \\\cline{1-2}
$\ell_q$ ($0<q<1$) & \multirow{2}{*}{$\mu|\beta|^{1/\gamma}$ }& \\
\cite{lq,log} & & \\\cline{1-2}
Log & \multirow{2}{*}{$\ds\mu {\log(\gamma|\beta| + 1)}/{\log(\gamma+1)}$ } & \\
\cite{log} & & \\\hline
  \end{tabular}
  \caption{Nonconvex penalty functions $\rho(|\beta|;\mu,\gamma)$ represented as in \eqref{eqn:ncpm}. The precise parametric representation is different than their original presentation but they are equivalent. We have taken care to normalize the different penalty functions so that $\mu$ is the sparsity parameter and $\gamma$ corresponds to the approximation of the indicator $I\{|\beta|>0\}$. For  SCAD, it is usually recommended to set $2\mu>3/\gamma^2$.
}
  \label{tab:ncp}%non convex penalties
\end{table*}

\subsection{Reformulating the problem \eqref{eqn:ncpm}}\label{ssec:ncpreform}

Let us begin by considering penalty functions $\rho$ of the form \eqref{eqn:pff} with $g$ a non-negative, increasing function of $|\beta|$. Observe that for any $\bb$ we can rewrite $\sum_{i=1}^p \min\{g(|\beta_i|),1\}$ as
\begin{align*}
&\min\left\{\sum_{i=1}^p g(|\beta_{(i)}|),1 + \sum_{i=2}^p g(|\beta_{(i)}|),\ldots, p-1 + g(|\beta_{(p)}|),p  \right\}\\
&=\min_{\ell\in\{0,\ldots,p\}} \left\{ \ell + \sum_{i>\ell}  g(|\beta_{(i)}|)  \right\}.
\end{align*}
It follows that \eqref{eqn:ncpm}
 can be rewritten \emph{exactly} as 
\begin{equation}\label{eqn:ncptna} % non-adjusted version 
\min_{\substack{\bb,\\\ell\in\{0,\ldots,p\}}} \left(L(\bb) + \mu\sum_{i>\ell} g(|\beta_{(i)}|) + \mu\ell\right)
\end{equation}
An immediate consequence is the following theorem:

\begin{theorem}\label{thm:MasT}
If $\bb^*$ is an optimal solution to \eqref{eqn:ncpm}, where $\rho(|\beta|;\mu,\gamma) = \mu\min\{g(|\beta|;\mu,\gamma),1\}$, then there exists some $\ell^*\in\{0,\ldots,p\}$ so that $\bb^*$ is optimal to its trimmed counterpart
\begin{equation*}%\label{eqn:counterpart}
\min_\bb L(\bb) + \mu\sum_{i>\ell^*} g(|\beta_{(i)}|).
\end{equation*}
In particular, the choice of $\ell^* = |\{i: g(|\beta_i^*|) \geq1 \}|$ suffices. %; further, if $g(|\beta|;\mu,\gamma)$ is strictly increasing in $|\beta|$, then this choice of $\ell^*$ is unique.
 Conversely, if $\bb^*$ is an optimal solution to \eqref{eqn:ncptna}, then $\bb^*$ in an optimal solution to \eqref{eqn:ncpm}.
\end{theorem}

It follows that the estimation problem \eqref{eqn:ncpm}, which decouples each loading $\beta_i$ in the penalty function, can be solved using ``trimmed'' estimation problems of the form \eqref{eqn:ncpt} with a trimmed penalty function that couples the loadings and only penalizes the $p-\ell^*$ smallest. Because the trimmed penalty function is generally nonconvex by nature, we will focus on comparing it with other nonconvex penalties for the remainder of the section.

\subsection{Trimmed reformulation examples}\label{ssec:ncpeg}

We now consider the structure of the estimation problem \eqref{eqn:ncpm} and the corresponding trimmed estimation problem for the clipped Lasso and MCP penalties. We use the $\ell_2^2$ loss throughout.

\subsubsection*{Clipped Lasso}

The clipped (or capped, or truncated) Lasso penalty \cite{cl,shen12} takes the component-wise form
$$\rho(|\beta|;\mu,\gamma) = \mu \min\{\gamma|\beta|,1\}.$$
Therefore, in our notation, $g$ is a multiple of the absolute value function. A plot of $\rho$ is shown in Figure \ref{fig:ncpa}. In this case, the estimation problem with $\ell_2^2$ loss is
\begin{equation}\label{eqn:clm}
\min_{\bb} \frac{1}{2}\|\y-\X\bb\|_2^2 + \mu\sum_i \min\{\gamma|\beta_i|,1\}.
\end{equation}
It follows that the corresponding trimmed estimation problem (\emph{c.f.} Theorem \ref{thm:MasT}) is exactly the trimmed Lasso problem studied earlier, namely,
\begin{equation}\label{eqn:clt}
\min_{\bb} \frac{1}{2}\|\y-\X\bb\|_2^2 + \mu\gamma \tk{\bb}.
\end{equation}
A distinct advantage of the trimmed Lasso formulation \eqref{eqn:clt} over the traditional clipped Lasso formulation \eqref{eqn:clm} is that it offers direct control over the desired level of sparsity vis-\`a-vis the discrete parameter $k$. We perform a deeper analysis of the two problems in Section \ref{ssec:ncpgenerality}.

\subsubsection*{MCP}

The MCP penalty takes the component-wise form
$$\rho(|\beta|;\mu,\gamma) = \mu \min\{g(|\beta|),1\}$$
where $g$ is any function with $g(|\beta|) = 2\gamma|\beta|-\gamma^2\beta^2$ whenever $|\beta| \leq1/\gamma$ and $g(|\beta|)\geq1$ whenever $|\beta| > 1/\gamma$. An example of one such $g$ is shown in Table \ref{tab:ncp}.
%$$
%g(|\beta|) = \left\{
%\begin{array}{cl}
%2|\beta|-\beta^2/\gamma& \text{ if } |\beta| \leq\gamma\\
%\geq \gamma & \text{ if } |\beta| > \gamma.
%\end{array}\right.
%$$
A plot of $\rho$ is shown in Figure \ref{fig:ncpa}. Another valid choice of $g$ is $g(|\beta|) = \max\{2\gamma|\beta|-\gamma^2\beta^2,\gamma|\beta|\}$. In this case, the trimmed counterpart is
\begin{equation*}%\label{eqn:clt}
\min_{\bb} \frac{1}{2}\|\y-\X\bb\|^2 + \mu\gamma \sum_{i>\ell}\max\left\{ 2|\beta_{(i)}| -\gamma\beta_{(i)}^2,|\beta_{(i)}|\right\}.
\end{equation*}

Note that this problem is amenable to the same class of techniques as applied to the trimmed Lasso problem in the form \eqref{eqn:clt} because of the increasing nature of $g$, although the subproblems with respect to $\bb$ are no longer convex (although it is a usual MCP estimation problem which is well-suited to convex optimization approaches; see \cite{sparsenet}).
Also observe that we can separate the penalty function into a trimmed Lasso component and another component:
$$\sum_{i>\ell} |\beta_{(i)}|\text{\quad and \quad} \sum_{i>\ell} \left(|\beta_{(i)}|-\gamma\beta_{(i)}^2\right)_+.$$
Observe that the second component is uniformly bounded above by $(p-\ell)/(4\gamma)$, and so as $\gamma\to\infty$, % and the MCP penalty approaches $\mu I\{|\beta|>0\}$,
the trimmed Lasso penalty dominates.

\begin{figure*}%mu = 3/2, gamma = 2/3 ???
\centering
\begin{subfigure}{0.49\linewidth}
        \centering
\begin{tikzpicture}
  \begin{axis}[ 
    xlabel=$|\beta|$,
    xmin = 0,
    xmax = 2,
    ymin = 0,
    ymax = 2,
    scale = .8,
    minor y tick num=0,
          yticklabels={,,,,$\mu$},
    minor x tick num=0,
          xticklabels={,0,,,$1/\gamma$},
          legend entries = {$\rho_\text{CL}$,$\rho_\text{MCP}$},
          legend style = {at={(.93,.07)}, anchor = south east},
  ] 
    \addplot[samples=\ns, color=blue, thick, dashed] { (x<1.5)*(x) + (x>1.5)*1.5};
    \addplot[samples=\ns, color=red, thick] { (x<1.5) * (2*x-x^2/1.5)  + (x>1.5)*(1.5) };
%    declare function={
%    func(\x)= (\x<1) * (x) + (\x>1)*(1.5)   ;
%  }
%    \addplot[samples=\ns] {func(x)};
  \end{axis}
\end{tikzpicture}
        \caption{Clipped Lasso and MCP}
        \label{fig:ncpa}
    \end{subfigure}
\begin{subfigure}{0.49\linewidth}
        \centering
 \begin{tikzpicture}
  \begin{axis}[ 
    xlabel=$|\beta|$,
    xmin = 0,
    xmax = 4,
    ymin = 0,
    ymax = 4,
    scale = .8,
    minor y tick num=0,
          yticklabels={,,,,$\mu$},
    minor x tick num=0,
          xticklabels={,0,,,$1$},
    legend entries = {$\rho_\text{log}$,$\rho_\text{$\ell_q$}$},
              legend style = {at={(.93,.07)}, anchor = south east},
  ] 
    \addplot[samples=\ns, color=blue, thick, dashed] { 3*ln(2*x/3+1)/ln(3) };
    \addplot[samples=\ns, color=red, thick] { 3*(x/3)^(1/2) };
\end{axis}
\end{tikzpicture}
\caption{Log and $\ell_q$}
\label{fig:ncpb}
\end{subfigure}\\[1ex]
%\begin{subfigure}{\linewidth}
%        \centering
% \begin{tikzpicture}
%  \begin{axis}[ 
%    xlabel=$|\beta|$,
%    %ylabel={$\rho_\text{SCAD}(|\beta|;\lambda,\gamma)$},
%    xmin = 0,
%    xmax = 4,
%    ymin = 0,
%    ymax = 4,
%    scale = .8,
%    minor y tick num=0,
%          yticklabels={,,,,$\mu$},
%    minor x tick num=0,
%          xticklabels={,0,$1/\gamma$,,$2\mu\gamma-1/\gamma$},
%    legend entries = {$\rho_\text{SCAD}$},
%          legend style = {at={(.93,.07)}, anchor = south east},
%  ] 
%    \addplot[samples=\ns, thick, black] { 3/2*((x<1 ) * (x) + (x>3)*(2)  + (x > 1)*(x<3)*( x^2-6*x+1 )/(4-8))}; 
%\end{axis}
%\end{tikzpicture}
%\caption{SCAD}
%\label{fig:ncpb}
%\end{subfigure}
\caption{Plots of $\rho(|\beta|;\mu,\gamma)$ for some of the penalty functions in Table \ref{tab:ncp}.}
\label{fig:ncp}
\end{figure*}
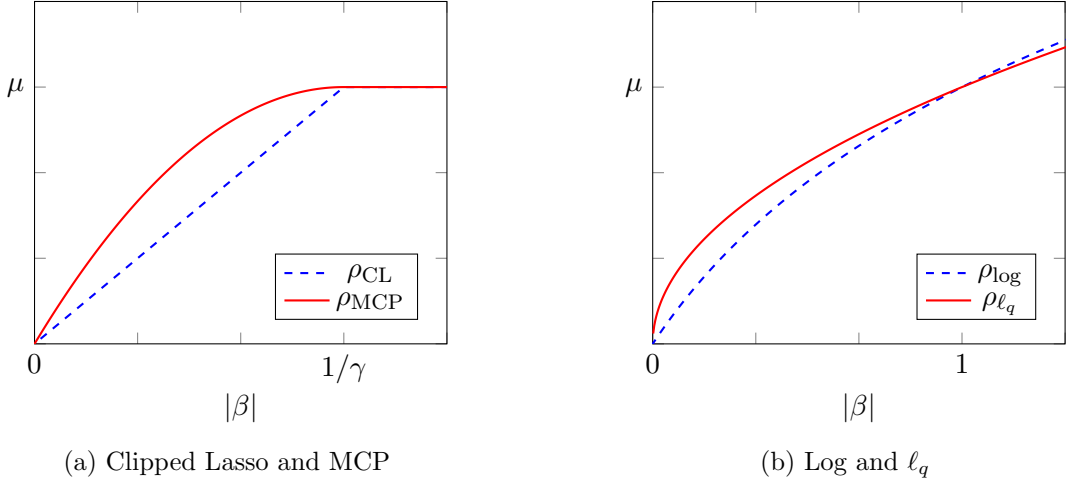

\subsection{The generality of trimmed estimation}\label{ssec:ncpgenerality}

We now turn our focus to more closely studying the relationship between the separable-penalty estimation problem \eqref{eqn:ncpm} and its trimmed estimation counterpart. The central problems of interest are the clipped Lasso and its trimmed counterpart, viz., the trimmed Lasso:\footnote{One may be concerned about the well-definedness of such problems (e.g. as guaranteed vis-\`a-vis coercivity of the objective, \emph{c.f.}\cite{rockafeller}). In all the results of Section \ref{ssec:ncpgenerality}, it is possible to add a regularizer $\eta\|\bb\|_1$ for some fixed $\eta>0$ to both $\cl$ and $\tl$ and the results remain valid, \emph{mutatis mutandis}. The addition of this regularizer implies coercivity of the objective functions and, consequently, that the minimum is indeed well-defined. For completeness, we note a technical reason for a choice of $\eta\|\bb\|_1$ is its positive homogeneity; thus, the proof technique of Lemma \ref{lem:key} easily adapts to this modification.}
\begin{align*}
&\ds \min_\bb \frac{1}{2}\|\y-\X\bb\|^2_2 + \mu \sum_i \min\{\gamma|\beta_i|,1\}\tag{$\textsc{CL}_{\mu,\gamma}$}\\
&\ds\min_\bb \frac{1}{2}\|\y-\X\bb\|^2_2 + \lambda \tka{\ell}{\bb}.\tag{$\textsc{TL}_{\lambda,\ell}$}
\end{align*}
%$$\begin{array}{lll}
%\ds\cl & &\ds \min_\bb \frac{1}{2}\|\y-\X\bb\|^2_2 + \mu \sum_i \min\{\gamma|\beta_i|,1\}\\
%\ds\tl & & \ds\min_\bb \frac{1}{2}\|\y-\X\bb\|^2_2 + \lambda \tka{\ell}{\bb}.
%\end{array}$$
As per Theorem \ref{thm:MasT}, if $\bb^*$ is an optimal solution to $\cl$, then $\bb^*$ is an optimal solution to $\tl$, where $\lambda=\mu\gamma$ and $\ell=|\{i:|\beta_i^*|\geq1/\gamma\}|$. We now consider the converse: given some $\lambda>0$ and $\ell\in\{0,1,\ldots,p\}$ and a solution $\bb^*$ to $\tl$, when does there exist some $\mu,\gamma>0$ so that $\bb^*$ is an optimal solution to $\cl$? As the following theorem suggests, the existence of such a $\gamma$ is closely connected to an underlying discrete form of ``convexity'' of the sequence of problems $\tla{\lambda,k}$ for $k\in\{0,1,\ldots,p\}$. We will focus on the case when $\lambda=\mu\gamma$, as this is the natural correspondence of parameters in light of Theorem \ref{thm:MasT}.

\begin{theorem}\label{thm:clconv}
If $\lambda>0$, $\ell\in\{0,\ldots,p\}$, and $\bb^*$ is an optimal solution to $\tl$, then there exist $\mu,\gamma>0$ with $\mu\gamma=\lambda$ and so that $\bb^*$ is an optimal solution to $\cl$ if and only if
\begin{equation}\label{eqn:clconv}
Z\tla{\lambda,\ell_e} < \frac{j-\ell_e}{j-i} Z\tla{\lambda,i} + \frac{\ell_e-i}{j-i} Z\tla{\lambda,j}
\end{equation}
for all $0\leq i< \ell_e < j \leq p$, where $Z(\textsc{P})$ denotes the optimal objective value to optimization problem $\textsc{(P)}$ and $\ell_e = \min\{\ell,\|\bb^*\|_0\}$.
\end{theorem}

Let us note why we refer to the condition in \eqref{eqn:clconv} as a discrete analogue of convexity of the sequence $\{z_k :=Z\tla{\lambda,k},\; k=0,\ldots,p\}$. In particular, observe that this sequence satisfies the condition of Theorem \ref{thm:clconv} if and only if the function defined as the linear interpolation between the points $(0,z_0)$, $(1,z_1)$, \ldots, and $(p,z_p)$ is strictly convex about the point $(\ell,z_\ell)$.\footnote{To be precise, we mean that the real-valued function that is a linear interpolation of the points has a subdifferential at the point $(\ell,z_\ell)$ which is an interval of strictly positive width.}

Before proceeding with the proof of the theorem, we state and prove a technical lemma about the structure of $\tl$.

\begin{lemma}\label{lem:key}
Fix $\lambda>0$ and suppose that $\bb^*$ is optimal to $\tl$.
\begin{enumerate}[(a)]
\item The optimal objective value of $\tl$ is $Z\tl = (\|\y\|_2^2-\|\X\bb^*\|_2^2)/2$.
\item If $\bb^*$ is also optimal to $\tla{\lambda,\ell'}$, where $\ell<\ell'$, then $\|\bb^*\|_0\leq\ell$ and $\bb^*$ is optimal to $\tla{\lambda,j}$ for all integral $j$ with $\ell<j<\ell'$.
\item If $\kappa:=\|\bb^*\|_0 <\ell$, then  $\bb^*$ is also optimal to $\tla{\lambda,\kappa}$, $\tla{\lambda,\kappa+1}$, \ldots, and $\tla{\lambda,\ell-1}$. Further, $\bb^*$ is \emph{not} optimal to $\tla{\lambda,0}$, $\tla{\lambda,1}$, \ldots, nor $\tla{\lambda,\kappa-1}$.
\end{enumerate}
\end{lemma}

\begin{proof}
Suppose $\bb^*$ is optimal to $\tl$. Define
$$a(\epsilon):= \|\y-\epsilon\X\bb^*\|_2^2/2 + \epsilon\lambda \tka{\ell}{\bb^*}.$$
By the optimality of $\bb^*$, $a(\epsilon)\geq a(1)$ for all $\epsilon\geq0$. As $a$ is a polynomial with degree at most two, one must have that $a'(1) = 0$. This implies that
$$a'(1) = -\langle\y,\X\bb^*\rangle + \|\X\bb^*\|_2^2 + \lambda\tka{\ell}{\bb^*} = 0.$$
Adding $(\|\y\|_2^2-\|\X\bb^*\|_2^2)/2$ to both sides, the desired result of part (a) follows.

Now suppose that $\bb^*$ is also optimal to $\tla{\lambda,\ell'}$, where $\ell'>\ell$. By part (a), one must necessarily have that $Z\tl=Z\tla{\lambda,\ell'} = (\|\y\|_2^2-\|\X\bb^*\|_2^2)/2$. Inspecting $Z\tl-Z\tla{\lambda,\ell'}$, we see that
$$0=Z\tl-Z\tla{\lambda,\ell'} = \lambda \sum_{i=\ell+1}^{\ell'} |\beta_{(i)}^*|.$$
Hence, $|\beta_{(\ell+1)}^*|=0$ and therefore $\|\bb^*\|_0\leq \ell$.

Finally, for any integral $j$ with $\ell\leq j\leq \ell'$, one always has that $Z\tl \geq Z\tla{\lambda,j} \geq Z\tla{\lambda,\ell'}$. As per the preceding argument, $Z\tl = Z\tla{\lambda,\ell}$ and so $Z\tl = Z\tla{\lambda,j}$, and therefore $\bb^*$ must also be optimal to $\tla{\lambda,j}$ by applying part (a). This completes part (b).

Part (c) follows from a straightforward inspection of objective functions and using the fact that $Z\tla{\lambda,j}\geq Z\tla{\lambda,\ell}$ whenever $j\leq\ell$.
\end{proof}

Using this lemma, we can now proceed with the proof of the theorem.

\begin{proof}[Proof of Theorem \ref{thm:clconv}]
Let $z_k = Z\tla{\lambda,k}$ for $k\in\{0,1,\ldots,p\}$. Suppose that $\mu,\gamma>0$ is so that $\lambda = \mu\gamma$ and $\bb^*$ is an optimal solution to $\cl$. Let $\ell_e = \min\{\ell,\|\bb^*\|_0\}$. Per equation \eqref{eqn:ncptna}, $\bb^*$ must be optimal to 
\begin{equation}\label{eqn:pf1}
\min_\bb \min_{k\in\{0,\ldots,p\}} \frac{1}{2}\|\y-\X\bb\|_2^2 + \mu k + \mu\gamma \tk{\bb}.
\end{equation}
Observe that this implies that if $k$ is such that $k$ is a minimizer of ${\min}_k \mu k + \mu\gamma \tk{\bb^*}$, then $\bb^*$ must be optimal to $\tla{\lambda,k}$.

We claim that this observation, combined with Lemma \ref{lem:key}, implies that
$$\ell_e=\underset{ {k\in\{0,\ldots,p\}} }{\argmin} \mu k + \mu\gamma \tk{\bb^*}.$$
This can be shown as follows:
\begin{enumerate}[(a)]
\item Suppose   $\ell\leq \|\bb^*\|_0$ and so $\ell_e = \min\{\ell,\|\bb^*\|_0\} = \ell$. Therefore, by Lemma \ref{lem:key}(b), $\bb^*$ is not optimal to $\tla{\lambda,j}$ for any $j< \ell$, and thus
$$\underset{ {k\in\{0,\ldots,p\}} }{\min} \mu k + \mu\gamma \tk{\bb^*} = \underset{ {k\in\{\ell,\ldots,p\}} }{\min} \mu k + \mu\gamma \tk{\bb^*}.$$

If $k>\ell$ is such that $k$ is a minimizer of ${\min}_k \mu k + \mu\gamma \tk{\bb^*}$, then $\bb^*$ must be optimal to $\tla{\lambda,k}$ (using the observation), and hence by Lemma \ref{lem:key}(b), $\|\bb^*\|_0\leq \ell$. Combined with $\ell\leq \|\bb^*\|_0$, this implies that $\|\bb^*\|_0=\ell$. Yet then,
$\mu\ell = \mu\ell + \mu\gamma \tka{\ell}{\bb^*} <  \mu k + \mu\gamma\tk{\bb^*}$, contradicting the optimality of $k$. Therefore, we conclude that $\ell_e=\ell$ is the \emph{only} minimizer of $\min_k \mu k + \mu\gamma\tk{\bb^*}$. 

\item Now instead suppose that $\ell_e = \|\bb^*\|_0 < \ell$. Lemma \ref{lem:key}(c) implies that any optimal solution $k$ to $\min_k \mu k + \mu\gamma \tk{\bb^*}$ must satisfy $k\geq\|\bb^*\|_0$ (by the second part combined with the observation). As before, if $k>\|\bb^*\|_0=\ell_e$, then $\mu k > \mu \ell_e$, and so $k$ cannot be optimal. As a result, $k=\ell_e=\|\bb^*\|_0$ is the unique minimum.

\end{enumerate}
In either case, we have that $\ell_e$ is the unique minimizer to $\min_k \mu k + \mu\gamma \tk{\bb^*}$.

It then follows that $Z(\text{problem }\eqref{eqn:pf1}) = z_{\ell_e} + \mu \ell_e$. Further, by optimality of $\bb^*$, $z_{\ell_e} + \mu  \ell_e < z_i + \mu  i$ for all $0\leq i\leq p$ with $i\neq\ell_e$. For $0\leq i < \ell_e$, this implies $\mu < (z_i-z_{\ell_e})/(\ell_e-i) $ and for $j>\ell_e$, $\mu > (z_{\ell_e}-z_j)/(j-\ell_e)$. In other words, for $0\leq i < \ell_e < j\leq p$,
$$\frac{z_{\ell_e}-z_j}{j-\ell_e} < \frac{z_i - z_{\ell_e}}{\ell_e-i}, \quad \text{i.e., }\;   z_{\ell_e} <\frac{j-\ell_e}{j-i}z_i + \frac{\ell_e-i}{j-i} z_j.$$
This completes the forward direction. The reverse follows in the same way by taking any $\mu$ with
\begin{equation*}
\mu\in \left( \max_{j> \ell_e} \frac{z_{\ell_e} -z_j}{j-\ell_e},  \min_{i<\ell_e} \frac{z_i-z_{\ell_e}}{\ell_e-i} \right).
\end{equation*}
\end{proof}

We briefly remark upon one implication of the proof of Theorem \ref{thm:clconv}. In particular, if $\bb^*$ is a solution to $\tl$ and $\ell <\|\bb^*\|_0$, then $\bb^*$ is not the solution to $\tla{\lambda,k}$ for any $k\neq \ell$.

An immediate question is whether the convexity condition \eqref{eqn:clconv} of Theorem \ref{thm:clconv} always holds. While the sequence $\{Z\tla{\lambda,k} : k=0,1,\ldots,p\}$ is always non-increasing, the following example shows that the convexity condition need not hold in general; as a result, there exist instances of the trimmed Lasso problem whose solutions \emph{cannot} be found by solving a clipped Lasso problem.

\begin{example}\label{eg:cl}
Consider the case when $p=n=2$ with
$$\y = \begin{pmatrix}1\\1\end{pmatrix} \text{ \quad and \quad } \mb X = \begin{pmatrix} 1 & -1\\-1&2\end{pmatrix}.$$
Let $\lambda =1/2$ and $ \ell = 1$, and consider $\min_{\bb} \|\y-\X\bb\|_2^2/2 + |\beta_{(2)}|/2 = \min_{\beta_1,\beta_2} (1-\beta_1+\beta_2)^2/2 + (1+\beta_1-2\beta_2)^2/2 + |\beta_{(2)}|/2.$
This has unique optimal solution $\bb^* = (3/2,1)$ with corresponding objective value $ z_1 = 3/4$. One can also compute $z_0 = Z\tla{1/2,0} = 39/40$ and $z_2 = Z\tla{1/2,2} = 0$. Note that $z_1 = 3/4 > (39/40)/2 + (0)/2 = z_0/2+z_2/2$, and so there do not exist any $\mu,\gamma>0$ with $\mu\gamma=1/2$ so that $\bb^*$ is an optimal solution to $\cla{\mu,\gamma}$ by Theorem \ref{thm:clconv}. Further, it is possible to show that $\bb^*$ is not an optimal solution to $\cl$ for \emph{any} choice of $\mu,\gamma\geq0$. (See Appendix \ref{app:proof}.)
\end{example}

An immediate corollary of this example, combined with Theorem \ref{thm:MasT}, is that the class of trimmed Lasso models contains the class of clipped Lasso models as a \emph{proper} subset, regardless of whether we restrict our attention to $\lambda=\mu\gamma$. In this sense, the trimmed Lasso models comprise a richer set of models. The relationship is depicted in stylized form in Figure \ref{fig:mc}.

\subsubsection*{Limit analysis}

It is important to contextualize the results of this section as $\lambda\to\infty$. This corresponds to $\gamma\to\infty$ for the clipped Lasso problem, in which case $\cl$ converges to the penalized form of subset selection:
\begin{equation*}
\min_\bb \frac{1}{2}\|\y-\X\bb\|_2^2 + \mu\|\bb\|_0.\tag{$\textsc{CL}_{\mu,\infty}$}
\end{equation*}
Note that penalized problems for all of the penalties listed in Table \ref{tab:ncp} have this as their limit as $\gamma\to\infty$.  On the other hand, $\tl$ converges to constrained best subset selection:
\begin{equation*}
\min_{\|\bb\|_0\leq \ell} \frac{1}{2} \|\y-\X\bb\|_2^2.\tag{$\textsc{TL}_{\infty,k}$}
\end{equation*}
Indeed, from this comparison it now becomes clear why a convexity condition of the form in Theorem \ref{thm:clconv} appears in describing when the clipped Lasso solves the trimmed Lasso problem. In particular, the conditions under $\cla{\mu,\infty}$ solves the constrained best subset selection problem $\tla{\infty,k}$ are precisely those in Theorem \ref{thm:clconv}.

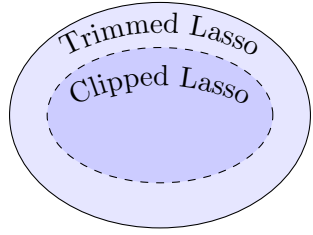
\begin{figure}
\centering
\begin{tikzpicture}
  \coordinate (c1) at (0,0);
      \pgfmathsetmacro{\colmixer}{mod(10*1,100)}%
      \path [draw, fill=blue!\colmixer, postaction=decorate] (c1) ellipse (2 and 1.5);
      \pgfmathsetmacro{\colmixer}{mod(10*2,100)}%
      \path [dashed, draw, fill=blue!\colmixer, postaction=decorate] (c1) ellipse (1.5 and .9) ;
 \path [ decoration={text along path, text={Clipped Lasso}, reverse path, text align={align=center}}, postaction=decorate] (2.5,-2.1) arc (0:180:2.5);
 \path [ decoration={text along path, text={Trimmed Lasso}, reverse path, text align={align=center}}, postaction=decorate] (2.6,-1.5) arc (0:180:2.6);
\end{tikzpicture}
\caption{Stylized relation of clipped Lasso and trimmed Lasso models. Every clipped Lasso model can be written as a trimmed Lasso model, but the reverse does not hold in general.}
\label{fig:mc}
\end{figure}

\subsection{Unbounded penalty functions}\label{ssec:ncpunbounded}

We close this section by now considering nonconvex penalty functions which are unbounded and therefore do not take the form $\mu\min\{g(|\beta|),1\}$. Two such examples are the $\ell_q$ penalty ($0<q<1$) and the log family of penalties as shown in Table \ref{tab:ncp} and depicted in Figure \ref{fig:ncpb}. Estimation problems with these penalties can be cast in the form
\begin{equation}\label{eqn:unbpm}
\min_{\ph} \frac{1}{2} \|\y-\X\ph\|_2^2 + \mu \sum_{i=1}^p g(|\phi_i|;\gamma)
\end{equation}
where $\mu,\gamma>0$ are parameters, $g$ is an unbounded and strictly increasing function, and $ g(|\phi_i|;\gamma) \xrightarrow{\gamma\to\infty} I\{|\phi_i|>0\}$. The change of variables in \eqref{eqn:unbpm} is intentional and its purpose will become clear shortly.

Observe that because $g$ is now unbounded, there exists some $\barl = \barl(\y,\X,\mu,\gamma)>0$ so that for all $\lambda>\barl$ any optimal solution $(\ph^*,\be^*)$ to the problem
\begin{equation}\label{eqn:unbpm-aux}
\min_{\ph,\be} \frac{1}{2} \|\y-\X(\ph+\be)\|_2^2 + \lambda\|\be\|_1+\mu \sum_{i=1}^p g(|\phi_i|;\gamma)
\end{equation}
has $\be^*=\mb0$.\footnote{The proof involves a straightforward modification of an argument along the lines of that given in Theorem \ref{thm:exactEquiv}. Also note that we can choose $\barl$ so that it is decreasing in $\gamma$, \emph{ceteris paribus}.} Therefore, \eqref{eqn:unbpm} is a special case of \eqref{eqn:unbpm-aux}. We claim that in the limit as $\gamma\to\infty$ (all else fixed), that \eqref{eqn:unbpm-aux} can be written exactly as a trimmed Lasso problem $\tla{\lambda,k}$ for some choice of $k$ and with the identification of variables $\bb = \ph+\be$.

We summarize this as follows:

\begin{proposition}
As $\gamma\to\infty$, the penalized estimation problem \eqref{eqn:unbpm} is a special case of the trimmed Lasso problem.
\end{proposition}
\begin{proof}
This can be shown in a straightforward manner: namely, as $\gamma\to\infty$, \eqref{eqn:unbpm-aux} becomes
\begin{equation*}%\label{eqn:unbpm-aux}
\min_{\ph,\be} \frac{1}{2} \|\y-\X(\ph+\be)\|_2^2 + \lambda\|\be\|_1+\mu \|\ph\|_0
\end{equation*}
which can be in turn written as 
\begin{equation*}%\label{eqn:unbpm-aux}
\min_{\substack{\ph,\be:\\\|\ph\|_0\leq k}} \frac{1}{2} \|\y-\X(\ph+\be)\|_2^2 + \lambda\|\be\|_1
\end{equation*}
for some $k\in\{0,1,\ldots,p\}$. But as per the observations of Section \ref{ssec:vardecomp}, this is exactly $\tla{\lambda,k}$ using a change of variables $\bb=\ph+\be$. In the case when $\lambda$ is sufficiently large, we necessarily have $\bb=\ph$ at optimality.
\end{proof}

While this result is not surprising (given that as $\gamma\to\infty$ the problem is \eqref{eqn:unbpm} is precisely penalized best subset selection), it is useful for illustrating the connection between \eqref{eqn:unbpm} and the trimmed Lasso problem even when the trimmed Lasso parameter $\lambda$ is not necessarily large: in particular, $\tla{\lambda,k}$ can be viewed as estimating $\bb$ as the sum of two components---a sparse component $\ph$ and small-norm (``noise'') component $\be$. Indeed, in this setup, $\lambda$ precisely controls the desirable level of allowed ``noise'' in $\bb$. From this intuitive perspective, it becomes clearer why the trimmed Lasso type approach represents a continuous connection between best subset selection ($\lambda$ large) and ordinary least squares ($\lambda$ small).

We close this section by making the following observation regarding problem \eqref{eqn:unbpm-aux}. In particular, observe that regardless of $\lambda$, we can rewrite this as
\begin{equation*}%\label{eqn:unbpm}
\min_{\bb} \frac{1}{2} \|\y-\X\bb\|_2^2 + \sum_{i=1}^p \widetilde{\rho}(|\beta_i|)
\end{equation*}
where $\widetilde{\rho}(|\beta_i|)$ is the new penalty function defined as
$$\widetilde{\rho}(|\beta_i|) = \min_{\phi+\epsilon = \beta_i} \lambda|\epsilon| + \mu g(|\phi|;\gamma).$$
For the unbounded and concave penalty functions shown in Table \ref{tab:ncp}, this new penalty function is quasi-concave and can be rewritten easily in closed form. For example, for the $\ell_q$ penalty $\rho(|\beta_i|) = \mu|\beta_i|^{1/\gamma}$ (where $\gamma>1$), the new penalty function is
$$\widetilde{\rho}(|\beta_i|) = \min\{\mu|\beta_i|^{1/\gamma},\lambda|\beta_i|\}.$$

%%%%%%%%%%%%%%%%%%%%%%%
%%% Algorithms
%%%%%%%%%%%%%%%%%%%%%%%

\section{Algorithmic Approaches}\label{sec:algs}

We now turn our attention to algorithms for estimation with the trimmed Lasso penalty. Our principle focus throughout will be the same problem considered in Theorem \ref{thm:exactEquiv}, namely
\begin{equation}\label{eqn:alg}
\ds\min_{\bb} \frac{1}{2}\|\y-\X\bb\|_2^2 + \lambda \tk{\bb} + \eta \|\bb\|_1
\end{equation}
We present three possible approaches to finding potential solutions to \eqref{eqn:alg}: a first-order-based alternating minimization scheme that has accompanying local optimality guarantees and was first studied in \cite{gotoh1,gotoh2}; an augmented Lagrangian approach that appears to perform noticeably better, despite lacking optimality guarantees; and a convex envelope approach. We contrast these methods with approaches for certifying global optimality of solutions to \eqref{eqn:alg} (described in \cite{thiao}) and include an illustrative computational example.
Implementations of the various algorithms presented can be found at
\begin{center}
\url{https://github.com/copenhaver/trimmedlasso}.
\end{center}

\subsection{Upper bounds via convex methods}\label{ssec:ub}

We start by focusing on the application of convex optimization methods to finding to finding potential solutions to \eqref{eqn:alg}. Technical details are contained in Appendix \ref{app:algsupp}.

\subsubsection*{Alternating minimization scheme}

We begin with a first-order-based approach for obtaining a locally optimal solution of \eqref{eqn:alg} as described in \cite{gotoh1,gotoh2}. The key tool in this approach is the theory of difference of convex optimization (``DCO'') \cite{anThesis,taoan97,dcSummary}. Set the following notation:
%We assume throughout that $k$, $\X$, $\y$, $\mu$, and $\lambda$ are fixed. Let
$$\begin{array}{lll}
f(\bb) &=& \|\y-\X\bb\|_2^2/2 + \lambda \tk{\bb}+ \eta \|\bb\|_1,\\
f_1(\bb) &=& \|\y-\X\bb\|_2^2/2 + (\eta+\lambda) \|\bb\|_1,\\
f_2(\bb) &=& \lambda \sum_{i=1}^k |\beta_{(i)}|.
\end{array}
$$
Let us make a few simple observations:
\begin{enumerate}[(a)]
\item Problem \eqref{eqn:alg} can be written as $\ds\min_\bb f(\bb)$.

\item For all $\bb$, $f(\bb) = f_1(\bb)-f_2(\bb)$.

\item The functions $f_1$ and $f_2$ are convex.
\end{enumerate}

While simple, these observations enable one to apply the theory of DCO, which focuses precisely on problems of the form
$$\min_{\bb} f_1(\bb)-f_2(\bb),$$
where $f_1$ and $f_2$ are convex. In particular, the optimality conditions for such a problem have been studied extensively \cite{dcSummary}. Let us note that while it may appear that the representation of the objective $f$ as $f_1-f_2$ might otherwise seem like an artificial algebraic manipulation, the min-min representation in Theorem \ref{thm:robeivInterp} shows how such a difference-of-convex representation can arise naturally.

We now discuss an associated alternating minimization scheme (or equivalently, a sequential linearization scheme), shown in Algorithm \ref{alg:1}, for finding local optima of \eqref{eqn:alg}. % by making use of Theorem \ref{thm:optCharac}.
The convergence properties of Algorithm \ref{alg:1} can be summarized as follows:\footnote{To be entirely correct, this result holds for Algorithm \ref{alg:1} with a minor technical modification---see details in Appendix \ref{app:algsupp}.}

\begin{theorem}[\cite{gotoh1}, Convergence of Algorithm \ref{alg:1}]\label{thm:altConvProp}
\begin{enumerate}[(a)]
\item The sequence $\{f(\bb^\ell):\ell=0,1,\ldots\}$, where $\bb^\ell$ are as found in Algorithm \ref{alg:1}, is non-increasing.

\item The set $\{\bg^\ell: \ell=0,1,\ldots\}$ is finite and eventually periodic.

\item Algorithm \ref{alg:1} converges in a finite number of iterations to local minimum of \eqref{eqn:alg}.

\item The rate of convergence of $f(\bb^\ell)$ is linear.

\end{enumerate}
\end{theorem}

\begin{algorithm}[ht]
\begin{enumerate}
\item Initialize with any $\bb^{0}\in\R^p$ ($\ell = 0$); for $\ell \geq 0$, repeat Steps 2-3 until $f(\bb^\ell) = f(\bb^{\ell+1})$.

\item Compute $\bg^\ell$ as
\begin{equation}\label{eqn:wrtbg}
\bg^\ell \in\begin{array}{ll}
\underset{\bg}{\operatorname{argmax}} & \langle \bg,\bb^\ell\rangle\\
\st& \ds\sum_i |\gamma_i| \leq \lambda k\\
& \ds|\gamma_i|\leq \lambda\;\forall i.
\end{array}
\end{equation}

\item Compute $\bb^{\ell+1}$ as 
\begin{equation}
\bb^{\ell+1} \in\underset{\bb}{\operatorname{argmin }} \;  \frac{1}{2}\|\y-\X\bb\|_2^2 +(\eta+\lambda)\|\bb\|_1 - \langle\bb,\bg^\ell\rangle.\label{eqn:wrtbb}
\end{equation}
%We assume, without loss of generality, that if $\bg^\ell = \bg^\kappa$ for some $\kappa<\ell$, then $\bb^{\ell+1} = \bb^{\kappa+1}$.
%\todot{need wlog cond?}

%
%Let $\tb$ be the final value of $\bb^\ell$.
%\item If $\sd f_2(\tb) \sub \sd f_1(\tb)$, halt. Otherwise, set $\bb^0 :=\tb$ 

\end{enumerate}
%\caption{Algorithm~\ref{algo:ao1-sk} for ($\textsc{CFA}_{q}$) }\label{algo:ao1-sk}
\caption{An alternating scheme for computing a local optimum to \eqref{eqn:alg}}\label{alg:1}
\end{algorithm}

\begin{obs}
Let us return to a remark that preceded Algorithm \ref{alg:1}. In particular, we noted that Algorithm \ref{alg:1} can also be viewed as a sequential linearization approach to solving \eqref{eqn:alg}. Namely, this corresponds to sequentially performing a linearization of $f_2$ (and leaving $f_1$ as is), and then solving the new convex linearized problem.

Further, let us note why we refer to Algorithm \ref{alg:1} as an alternating minimization scheme. In particular, in light of the reformulation \eqref{eqn:mainReform} of \eqref{eqn:alg}, we can rewrite \eqref{eqn:alg} exactly as 
$$\eqref{eqn:alg} = \begin{array}{ll}
\ds\min_{\bb,\bg} & f_1(\bb) - \langle\bg,\bb\rangle\\
\st& \ds\sum_i |\gamma_i| \leq \lambda k\\
& \ds|\gamma_i|\leq \lambda\;\forall i.
\end{array}$$
In this sense, if one takes care in performing alternating minimization in $\bb$ (with $\bg$ fixed) and in $\bg$ (with $\bb$ fixed) (as in Algorithm \ref{alg:1}), then a locally optimal solution is guaranteed.
\end{obs}

We now turn to how to actually apply Algorithm \ref{alg:1}. Observe that the algorithm is quite simple; in particular, it only requires solving two types of well-structured convex optimization problems. The first such problem, for a fixed $\bb$, is shown in \eqref{eqn:wrtbg}. This can be solved in closed form by simply sorting the entries of $|\bb|$, i.e., by finding $|\beta_{(1)}|,\ldots,|\beta_{(p)}|$. 
The second subproblem, shown in \eqref{eqn:wrtbb} for a fixed $\bg$, is precisely the usual Lasso problem and is amenable to any of the possible algorithms for the Lasso \cite{tibshirani,lars,hastie}.

\subsubsection*{Augmented Lagrangian approach}

We briefly mention another technique for finding potential solutions to \eqref{eqn:alg} using an Alternating Directions Method of Multiplers (ADMM) \cite{admm} approach. To our knowledge, the application of ADMM to the trimmed Lasso problem is novel, although it appears closely related to \cite{admmteng}. 
We begin by observing that \eqref{eqn:alg} can be written exactly as
$$\begin{array}{ll}
\ds\min_{\bb,\bg}& \frac{1}{2}\left\|\y-\X\bb\right\|_2^2 + \eta\left\|\bb\right\|_1 + \lambda \tk{\bg} \\
\st & \bb=\bg,
\end{array}$$
which makes use of the canonical variable splitting. Introducing dual variable $\mb q\in\R^p$ and parameter $\sigma>0$, this becomes in augmented Lagrangian form
\begin{align}
\ds\min_{\bb,\bg} \max_{\mb q}\;&\frac{1}{2}\left\|\y-\X\bb\right\|_2^2 + \eta\left\|\bb\right\|_1 +\lambda \tk{\bg} + \nonumber\\
& \langle \mb q, \bb-\bg\rangle + \frac{\sigma}{2}\left\|\bb-\bg\right\|_2^2.\label{eqn:admm}
\end{align}

The utility of such a reformulation is that it is directly amenable to ADMM, as detailed in Algorithm \ref{alg:admm}. While the problem is nonconvex and therefore the ADMM is not guaranteed to converge, numerical experiments suggest that this approach has superior performance to the DCO-inspired method considered in Algorithm \ref{alg:1}.% (Indeed, the ADMM does not necessarily converge even in the orthogonal design case.)

We close by commenting on the subproblems that must be solved in Algorithm \ref{alg:admm}. Step 2 can be carried out using ``hot'' starts. Step 3 is the solution of the trimmed Lasso in the orthogonal design case and can be solved by performed by sorting $p$ numbers; see Appendix \ref{app:algsupp}.

\begin{algorithm}[ht]
\begin{enumerate}
\item Initialize with any $\bb^0,\bg^0,\mb q^0 \in\R^p$ and $\sigma>0$. Repeat, for $\ell\geq 0$,  Steps 2, 3, and 4 until a desired numerical
convergence tolerance is satisfied.

\item Set
\begin{align*}
\bb^{\ell+1} \in  \ds\underset{\bb}{\operatorname{argmin}}\; &\frac{1}{2}\|\y-\X\bb\|_2^2+\eta\|\bb\|_1\;+\\
&\langle \mb q^\ell, \bb\rangle + \frac{\sigma}{2}\|\bb-\bg^\ell\|_2^2.
\end{align*}

\item Set
$$\bg^{\ell+1} \in\underset{\bg}{\operatorname{argmin} } \;\lambda \tk{\bg} + \frac{\sigma}{2}\|\bb^{\ell+1}-\bg\|_2^2  - \langle \mb q^\ell,\bg\rangle.$$

\item Set $\mb q^{\ell+1}  = \mb q^{\ell} + \sigma\left(\bb^{\ell+1} - \bg^{\ell+1}\right)$. 

\end{enumerate}
%\caption{Algorithm~\ref{algo:ao1-sk} for ($\textsc{CFA}_{q}$) }\label{algo:ao1-sk}
\caption{ADMM algorithm for \eqref{eqn:admm}}\label{alg:admm}
\end{algorithm}

\subsubsection*{Convexification approach}\label{ssec:convenv}

We briefly consider the convex relaxation of the problem \eqref{eqn:alg}. We begin by computing the convex envelope \cite{rockafeller,BV2004} of $\tke$ on $[-1,1]^p$ (here the choice of $[-1,1]^p$ is standard, such as in the convexification of $\ell_0$ over this set which leads to $\ell_1$). The proof follows standard techniques (e.g. computing the biconjugate\cite{rockafeller}) and is omitted.

\begin{lemma}\label{lem:convenv}
The convex envelope of $\tke$ on $[-1,1]^p$ is the function $\overline{\tke}$ defined as
$$\overline{\tke}(\bb) = \left(\|\bb\|_1-k\right)_+.$$
\end{lemma}

In words, the convex envelope of $\tke$ is a ``soft thresholded'' version of the Lasso penalty (thresholded at level $k$). This can be thought of as an alternative way of interpreting the name ``trimmed Lasso.''

As a result of Lemma \ref{lem:convenv}, it follows that the convex analogue of \eqref{eqn:alg}, as taken over $[-1,1]^p$, is precisely
\begin{equation}\label{eqn:convenv}
\min_{\bb} \frac{1}{2}\|\y-\X\bb\|_2^2 + \eta\|\bb\|_1 + \lambda\left(\|\bb\|_1-k\right)_+.
\end{equation}
Problem \eqref{eqn:convenv} is amenable to a variety of convex optimization techniques such as subgradient descent \cite{BV2004}.

\subsection{Certificates of optimality for \eqref{eqn:alg}}\label{ssec:mio}

We close our discussion of the algorithmic implications of the trimmed Lasso by discussing techniques for finding certifiably optimal solutions to \eqref{eqn:alg}. All approaches presented in the preceding section find potential candidates for solutions to \eqref{eqn:alg}, but none is necessarily globally optimal. Let us return to a representation of \eqref{eqn:alg} that makes use Lemma \ref{lemma:miprep}:
\begin{equation*}%\label{eqn:mainReform}
\begin{array}{ll}
\ds\min_{\bb,\zz} & \frac{1}{2}\|\y-\X\bb\|_2^2 + \eta\|\bb\|_1 + \lambda\langle\zz,|\bb|\rangle\\
\st& \ds\sum_i z_i=p-k\\
&\zz\in\{0,1\}^p.
\end{array}
\end{equation*}
As noted in \cite{gotoh1}, this representation of \eqref{eqn:alg} is amenable to mixed integer optimization (``MIO'') methods \cite{bonami} for finding globally optimal solutions to \eqref{eqn:alg}, in the same spirit as other MIO-based approaches to statistical problems \cite{bmlqs,bkm}.

One approach, as described in \cite{thiao}, uses the notion of ``big $M$.'' In particular, for $M>0$ sufficiently large, problem \eqref{eqn:alg} can be written exactly as the following linear MIO problem:

\begin{equation}\label{eqn:pf7}
\begin{array}{ll}
\ds\min_{\bb,\zz,\mb a} &\ds \frac{1}{2} \|\y-\X\bb\|_2^2 + \eta\|\bb\|_1 + \lambda\sum_i a_i\\
\st& \ds\sum_i z_i=p-k\\
&\zz\in\{0,1\}^p\\
& \mb a \geq \bb + M\zz- M\mb 1\\
& \mb a \geq -\bb + M\zz- M\mb 1\\
& \mb a\geq \mb0.
\end{array}
\end{equation}
This representation as a linear MIO problem enables the direct application of numerous existing MIO algorithms (such as \cite{gurobi}).\footnote{There are certainly other possible representations of \eqref{eqn:mainReform}, such as using special ordered set (SOS) constraints, see e.g. \cite{bkm}. Without more sophisticated tuning of $M$ as in \cite{bkm}, the SOS formulations appear to be vastly superior in terms of time required to prove optimality. The precise formulation essentially takes the form of problem \eqref{eqn:BSSr}. An SOS-based implementation is provided in the supplementary code as the default method of certifying optimality.} Also, let us note that the linear relaxation of \eqref{eqn:pf7}, i.e., problem \eqref{eqn:pf7} with the constraint $\zz\in\{0,1\}^p$ replaced with $\zz\in[0,1]^p$, is the problem
\begin{equation*}%\label{eqn:pf7}
\ds\min_{\bb}  \frac{1}{2} \|\y-\X\bb\|_2^2 + \eta\|\bb\|_1 + \lambda \left(\|\bb\|_1-Mk\right)_+,
\end{equation*}
where we see the convex envelope penalty appear directly. As such, when $M$ is large, the linear relaxation of \eqref{eqn:pf7} is the ordinary Lasso problem $\min_{\bb}  \frac{1}{2} \|\y-\X\bb\|_2^2 + \eta\|\bb\|_1$.

\subsection{Computational example}

Because a rigorous computational comparison is not the primary focus of this paper, we provide a limited demonstration that describes the behavior of solutions to \eqref{eqn:alg} as computed via the different approaches. Precise computational details are contained in Appendix \ref{app:compdetail}. We will focus on two different aspects: sparsity and approximation quality.

\subsubsection*{Sparsity properties}

As the motivation for the trimmed Lasso is ostensibly sparse modeling, its sparsity properties  are  of particular interest. We consider a problem instance with $p=20$, $n=100$, $k=2$, and signal-to-noise ratio 10 (the sparsity of the ground truth model $\bb_\text{true}$ is $10$). The relevant coefficient profiles as a function of $\lambda$ are shown in Figure \ref{fig:coeffpath}. In this example none of the convex approaches finds the optimal two variable solution computed using mixed integer optimization. Further, as one would expect \emph{a priori}, the optimal coefficient profiles (as well as the ADMM profiles) are not continuous in $\lambda$. Finally, note that by design of the algorithms, the alternating minimization and ADMM approaches yield solutions with sparsity at most $k$ for $\lambda$ sufficiently large.

\begin{figure}
\centering
\includegraphics[width=.8\textwidth]{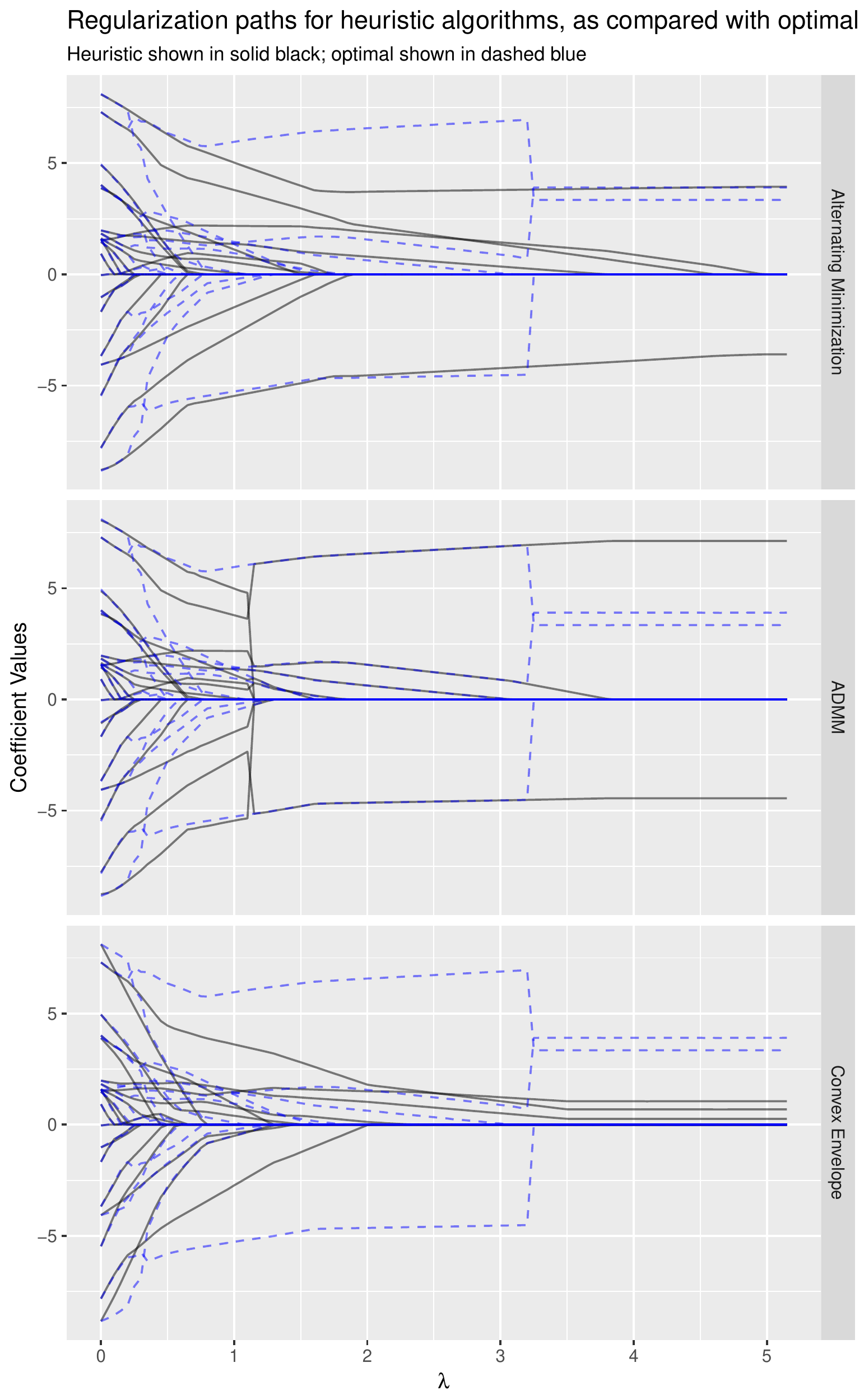}
\caption{}
\label{fig:coeffpath}
\end{figure}

\subsubsection*{Optimality gap}

Another critical question is the degree of suboptimality of solutions found via the convex approaches. We average optimality gaps across 100 problem instances with $p=20$, $n=100$, and $k=2$; the relevant results are shown in Figure \ref{fig:optgap}. The results are entirely as one might expect. When $\lambda$ is small and the problem is convex or nearly convex, the heuristics perform well. However, this breaks down as $\lambda$ increases and the sparsity-inducing nature of the trimmed Lasso penalty comes into play. Further, we see that the convex envelope approach tends to perform the worst, with the ADMM performing the best of the three heuristics. This is perhaps not surprising, as any solution found via the ADMM can be guaranteed to be locally optimal by subsequently applying the alternating minimization scheme of Algorithm \ref{alg:1} to any solution found via Algorithm \ref{alg:admm}.

\begin{figure}
\centering
\includegraphics[width=.8\textwidth]{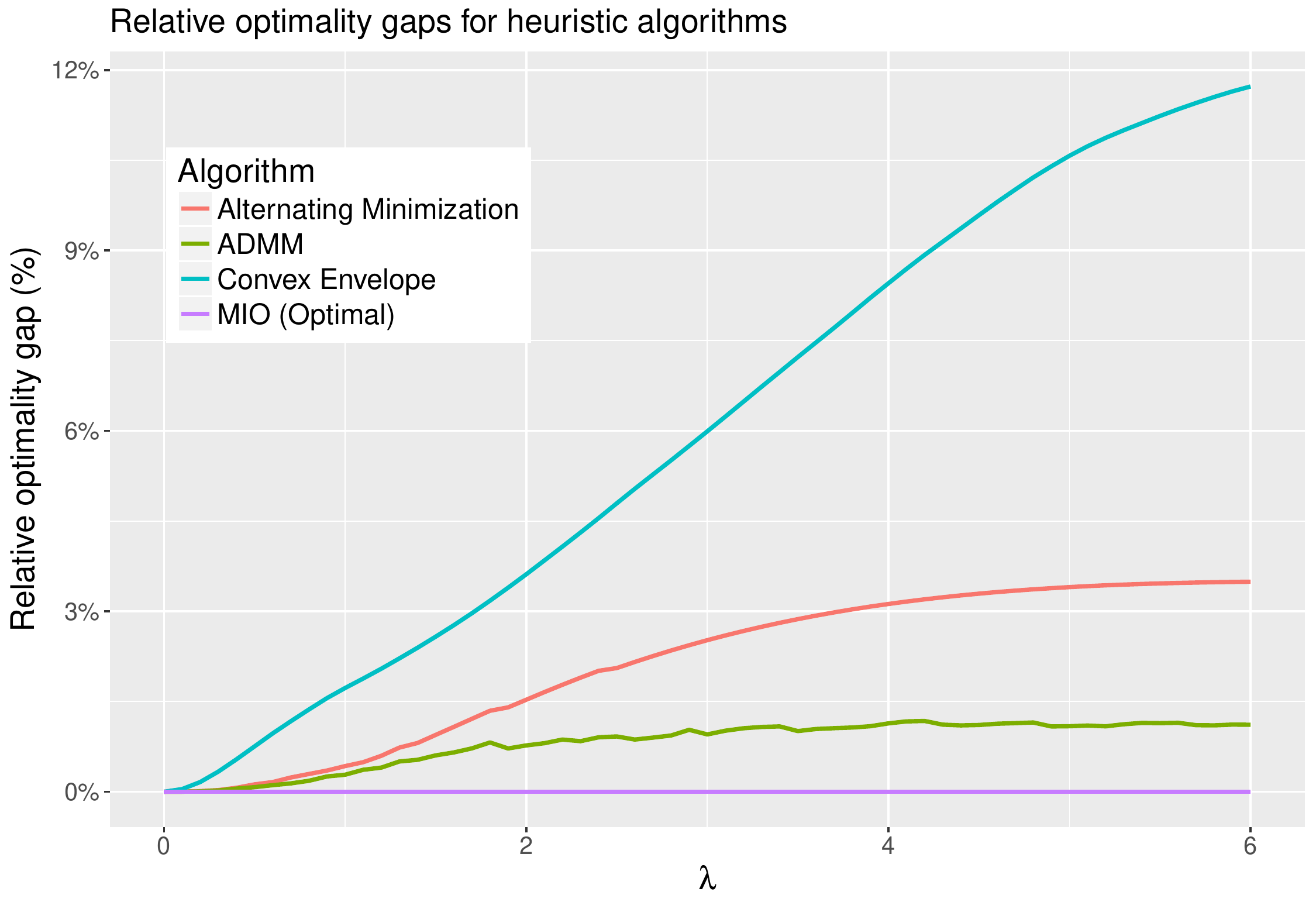}
\caption{}
\label{fig:optgap}
\end{figure}

\subsubsection*{Computational burden}

Loosely speaking, the heuristic approaches all carry a similar computational cost per iteration, namely, solving a Lasso-like problem. In contrast, the MIO approach can take significantly more computational resources. However, by design, the MIO approach maintains a suboptimality gap throughout computation and can therefore be terminated, before optimality is certified, with a certificate of suboptimality. We do not consider any empirical analysis of runtime here.

\subsubsection*{Other considerations}

There are other additional computational considerations that are potentially of interest as well, but they are primarily beyond the scope of the present work. For example, instead of considering optimality purely in terms of objective values in \eqref{eqn:alg}, there are other critical notions from a statistical perspective (e.g. ability to recover true sparse models and performance on out-of-sample data) that would also be necessary to consider across the multiple approaches.

%%%%%%%%%%%%
%%% Conclusions %%%
%%%%%%%%%%%%

\section{Conclusions}\label{sec:conc}

In this work, we have studied the trimmed Lasso, a nonconvex adaptation of Lasso that acts as an exact penalty method for best subset selection. Unlike some other approaches to exact penalization which use coordinate-wise separable functions, the trimmed Lasso offers direct control of the desired sparsity $k$. Further, we emphasized the interpretation of the trimmed Lasso from the perspective of robustness. In doing so, we provided contrasts with the SLOPE penalty as well as comparisons with estimators from the robust statistics and total least squares literature.

We have also taken care to contextualize the trimmed Lasso within the literature on nonconvex penalized estimation approaches to sparse modeling, showing that penalties like the trimmed Lasso can be viewed as a generalization of such approaches in the case when the penalty function is bounded. In doing so, we also highlighted how precisely the problems were related, with a complete characterization given in the case of the clipped Lasso.

Finally, we have shown how modern developments in optimization can be brought to bear for the trimmed Lasso to create convex optimization optimization algorithms that can take advantage of the significant developments in algorithms for Lasso-like problems in recent years.

Our work here raises many interesting questions about further properties of the trimmed Lasso and the application of similar ideas in other settings. We see two particularly noteworthy directions of focus: algorithms and statistical properties. For the former, we anticipate that an approach like trimmed Lasso, which leads to relatively straightforward algorithms that use close analogues from convex optimization, is  simple to interpret and to implement. At the same time, the heuristic approaches to the trimmed Lasso presented herein carry no more of a computational burden than solving convex, Lasso-like problems. On the latter front, we anticipate that a deeper analysis of the statistical properties of estimators attained using the trimmed Lasso would help to illuminate it in its own right while also further connecting it to existing approaches in the statistical estimation literature.

%%%%%%%%%%%%%%%%%%
%%% Appendices                   %%%
%%%%%%%%%%%%%%%%%%

\begin{appendices}

\section{General min-max representation of SLOPE}\label{app:slope}

For completeness, in this appendix we include the more general representation of the SLOPE penalty $R_{\textsc{SLOPE}(\mb w)}$ in the same spirit of Proposition \ref{prop:slope}. Here we work with SLOPE in its most general form, namely,
$$R_{\textsc{SLOPE}(\mb w)}(\bb) = \sum_{i=1}^p w_i |\beta_{(i)}|, $$
where $\mb w$ is a (fixed) vector of weights with $w_1\geq w_2\geq \cdots\geq w_p\geq 0$ and $w_1>0$.

To describe the general min-max representation, we first set some notation. For a matrix $\D\in\R^{n\times p}$, we let $\bs\nu(\D)\in\R^p$ be the vector $(\|\D_1\|_2,\ldots,\|\D_p\|_2)$ with entries sorted so that $\nu_1\geq \nu_2\geq \cdots \geq \nu_p$. As usual, for two vectors $\x$ and $\y$, we use $\x\leq\y$ to denote that coordinate-wise inequality holds. With this notation, we have the following:
\begin{proposition}\label{prop:slopefullgenerality}
Problem \eqref{eqn:roprimitive} with uncertainty set
$$\U_\w =\left\{\D : \bs\nu(\D) \leq \mb w \right\}
$$
is equivalent to problem \eqref{eqn:a1} with $R(\bb)=R_{\textsc{SLOPE}(\mb w)}(\bb) $. Further,
problem \eqref{eqn:roprimitive} with uncertainty set
$$\U_{\w} =\left\{\D : \|\D\ph\|_2 
\leq R_{\textsc{SLOPE}(\w)}(\ph) \;\forall \ph \right\}
$$
is equivalent to problem \eqref{eqn:a1} with $R(\bb)=R_{\textsc{SLOPE}(\mb w)}(\bb) $.
\end{proposition}

The proof, like the proof of Proposition \ref{prop:slope}, follows basic techniques described in \cite{RObook} and is therefore omitted.

\section{Additional proofs}\label{app:proof}

This appendix section contains supplemental proofs not contained in the main text.

\begin{proof}[Proof of Theorem \ref{thm:exactEquiv}]
Let $\barl = \|\y\|_2\cdot\left(\max_j\|\x_j\|_2\right)$, where $\x_j$ denotes the $j$th row of $\X$. We fix $\lambda>\barl$, $k$, and $\eta>0$ throughout the entire proof. We begin by observing that it suffices to show that any solution $\bb$ to
\begin{equation}\label{eqn:thmmain}
\ds\min_{\bb} \frac{1}{2}\|\y-\X\bb\|_2^2 + \lambda \tk{\bb} + \eta \|\bb\|_1
\end{equation}
satisfies $\tk{\bb} = 0$, or equivalently, $\|\bb\|_0\leq k$. As per Lemma \ref{lemma:miprep}, problem \eqref{eqn:thmmain} can be rewritten exactly as
\begin{equation}\label{eqn:mainReform}
\begin{array}{ll}
\ds\min_{\bb,\zz} & \frac{1}{2}\|\y-\X\bb\|_2^2 + \lambda\langle\zz,|\bb|\rangle+ \eta\|\bb\|_1 \\
\st& \ds\sum_i z_i=p-k\\
&\zz\in\{0,1\}^p.
\end{array}
\end{equation}
Let $(\bb^*,\zz^*)$ be any solution to \eqref{eqn:mainReform}. Observe that necessarily $\bb^*$ is also a solution to the problem
\begin{equation}\label{eqn:pfsupp1}
\min_{\bb}\frac{1}{2}\|\y-\X\bb\|_2^2  + \lambda\langle\zz^*,|\bb|\rangle+ \eta\|\bb\|_1.
\end{equation}
Note that, unlike \eqref{eqn:thmmain}, the problem in \eqref{eqn:pfsupp1} is readily amenable to an analysis using the theory of proximal gradient methods \cite{combetteswasj,bauschke}. In particular, we must have for any $\gamma>0$ that 
\begin{equation}\label{eqn:pfsupp2}
\bb^* = \prox_{\gamma R} \left(\bb^* - \gamma(\X'\X\bb^* - \X'\y)\right),
\end{equation}
where $\ds R(\bb) = \eta\|\bb\|_1 + \lambda \sum_{i\;:\:z_i^*=1} |\beta_i|$. Suppose that $\tk{\bb^*}>0$. In particular, for some $j\in\{1,\ldots,p\}$, we have $\beta_j^* \neq 0$ and $z_j^*=1$. Yet, as per \eqref{eqn:pfsupp2},\footnote{This is valid for the following reason: since $\beta_j^*\neq 0$ and $\beta_j^*$ satisfies \eqref{eqn:pfsupp2}, it must be the case that $\left|\beta_j^* - \gamma \x_{j}'(\X\bb^* - \y)\right| > \gamma(\eta+\lambda)$, for otherwise the soft-thresholding operator at level $\gamma(\eta+\lambda)$ would set this quantity to zero.}
$$\left|\beta_j^* - \gamma \langle\x_{j},\X\bb^* - \y\rangle\right| > \gamma(\eta+\lambda)\;\quad \text{ for all } \gamma>0,$$
where $\x_j$ denotes the $j$th row of $\X$. This implies that
$$\left|\langle\x_j,\X\bb^*-\y\rangle\right| \geq \eta+\lambda.$$
Now, using the definition of $\barl$, observe that
\begin{align*}
\eta+\lambda\leq \left|\langle\x_j,\X\bb^*-\y\rangle\right|& \leq \|\x_j\|_2 \|\X\bb^*-\y\|_2\\
&\leq \|\x_j\|_2 \|\y\| \leq \barl<\lambda,
\end{align*}
which is a contradiction since $\eta>0$. Hence, $\tk{\bb^*}=0$, completing the proof.
\end{proof}

\subsubsection*{Extended statement of Proposition \ref{prop:asymp}}

We now include a precise version of the convergence claim in Proposition \ref{prop:asymp}. Let us set a standard notion: we say that $\bb$ is $\epsilon$-optimal (for $\epsilon>0$)  to an optimization problem $(\textrm{P})$ if the optimal objective value of $(\textrm{P})$ is within $\epsilon$ of the objective value of $\bb$. We add an additional regularizer $\eta\|\bb\|_1$, for $\eta>0$ fixed, to the objective in order to ensure coercivity of the objective functions.

\begin{proposition}[Extended form of Proposition \ref{prop:asymp}]
Let $g:\R_+\to\R_+$ be an unbounded, continuous, and strictly increasing function with $g(0)=0$. 
Consider the problems
\[\label{eqn:cvg1}
\ds\min_{\bb} \frac{1}{2}\|\y-\X\bb\|_2^2 + \lambda\pi_k^g(\bb) + \eta\|\bb\|_1
\]
and
\[\label{eqn:cvg2}
\ds\min_{\|\bb\|_0\leq k} \frac{1}{2}\|\y-\X\bb\|_2^2 + \eta\|\bb\|_1.
\]
For every $\epsilon>0$, there exists some $\ubarl=\ubarl(\epsilon)>0$ so that for all $\lambda>\ubarl$,
\begin{enumerate}
\item For every optimal $\bbs$ to \eqref{eqn:cvg1}, there is some $\whbb$ so that $\|\bbs-\whbb\|_2\leq \epsilon$, $\whbb$ is feasible to \eqref{eqn:cvg2}, and $\whbb$ is $\epsilon$-optimal to \eqref{eqn:cvg2}.
\item Every optimal $\bb^*$ to \eqref{eqn:cvg2} is $\epsilon$-optimal to \eqref{eqn:cvg1}.
\end{enumerate}
\end{proposition}

\begin{proof}
The proof follows a basic continuity argument that is simpler than the one presented below in Theorem \ref{thm:corprecise}. For that reason, we do not include a full proof. Observe that the assumptions on $g$ imply that $g^{-1}$ is well-defined on, say, $g([0,1])$. If we let $\epsilon>0$ and suppose that $\bbs$ is optimal to \eqref{eqn:cvg1}, where $\lambda > \ubarl := \|\y\|_2^2/(2g(\epsilon/p))$, and if we define $\whbb$ to be $\bbs$ with all but the $k$ largest magnitude entries truncated to zero (ties broken arbitrarily), then $\pi_k^g(\bbs)\leq \|\y\|_2^2/(2\lambda)$ and $\pi_k^g (\bbs) = \sum_{i=1}^p g(|\beta_i^*-\widehat{\beta}_i|)$ so that $|\beta_i^*-\widehat{\beta}_i| \leq g^{-1}(\|\y\|_2^2/(2\lambda)) \leq \epsilon/p$ by definition of $\ubarl$. Hence, $\|\bbs-\whbb\|_1\leq \epsilon$, and all the other claims essentially follow from this.
\end{proof}

\begin{proof}[Proof of Theorem \ref{thm:robeivInterp}]
We begin by showing that for any $\bb$,
\begin{equation*}%\label{eqn:pfsupp3}
\min_{\D\in\U_k^\lambda}\|\y-(\X+\D)\bb\|_2 =  \left(\|\y-\X\bb\|_2 - \lambda \sum_{i=1}^k |\beta_{(i)}|   \right)_+
\end{equation*}
where $(a)_+:=\max\{0,a\}$. Fix $\bb$ and set $\mb r = \y-\X\bb$. We assume without loss of generality that $\mb r\neq \mb0$ and that $\bb\neq\mb0$. For any $\D$, note that $\|\mb r - \D\bb\|_2 \geq 0$ and $\|\mb r - \D\bb\|_2 \geq \|\mb r\|_2 - \|\D\bb\|_2$ by the reverse triangle inequality. Now observe that for $\D\in\U_k^\lambda$,
$$\|\D\bb\|_2  \leq \sum_i |\beta_i| \|\D_i\|_2 \leq \sum_{i=1}^k \lambda |\beta_{(i)}|.$$
Therefore, $\|\mb r-\D\bb\|_2 \geq \left(\|\mb r\|_2 - \lambda\sum_{i=1}^k |\beta_{(i)}| \right)_+$. Let $I\sub\{1,\ldots,p\}$ be a set of $k$ indices which correspond to the $k$ largest entries of $\bb$ (if $|\beta_{(k)}|=|\beta_{(k+1)}|$, break ties arbitrarily). Define $\D\in\U_k^\lambda$ as the matrix whose $i$th column is
$$\left\{\begin{array}{rl}
\underline{\lambda}\sgn(\beta_i)\mb r / \|\mb r\|_2,&i\in I\\
0,&i\notin I,
\end{array}\right.$$
where $\underline{\lambda} = \min\left\{\lambda, \|\mb r\|_2/\left(\sum_{i=1}^k|\beta_{(i)}|\right)\right\}$. It is easy to verify that $\D\in\U_k^\lambda$ and  $\|\mb r-\D\bb\|_2 = \left(\|\mb r\|_2 - \lambda\sum_{i=1}^k |\beta_{(i)}| \right)_+$. 
Combined with the lower bound, we have
$$\min_{\D\in\U_k^\lambda} \|\y-(\X+\D)\bb\|_2 = \left(\|\y-\X\bb\|_2 -\lambda\sum_{i=1}^k |\beta_{(i)}|\right)_+$$
which completes the first claim.

It follows that the problem \eqref{eqn:eivconhompen} can be rewritten exactly as
\begin{equation}\label{eqn:pfsupp3}
\min_{\bb} \left(\|\y-\X\bb\|_2 - \lambda \sum_{i=1}^k |\beta_{(i)}|   \right)_+ + r(\bb).
\end{equation}

To finish the proof of the theorem, it suffices to show that if $\bb^*$ is a solution to \eqref{eqn:pfsupp3}, then
$$\|\y-\X\bb^*\|_2 - \lambda\sum_{i=1}^k |\beta_{(i)}^*| \geq 0.$$
If this is not true, then $\|\y-\X\bb^*\|_2 - \lambda\sum_{i=1}^k |\beta_{(i)}^*| <0$ and so $\bb^*\neq \mb 0$. However, this implies that for $1>\epsilon>0$ sufficiently small, $\bb_\epsilon:=(1-\epsilon)\bb^*$ satisfies $\|\y-\X\bb_\epsilon\|_2 - \lambda\sum_{i=1}^k |(\beta_\epsilon)_{(i)}| <0$. This in turn implies that
$$\begin{array}{l}
\left(\|\y-\X\bb_\epsilon\|_2 - \lambda\sum_{i=1}^k |(\beta_\epsilon)_{(i)}| \right)_+ + r(\bb_\epsilon)\\
<  \left(\|\y-\X\bb^*\|_2 - \lambda\sum_{i=1}^k |\beta_{(i)}^*| \right)_+ +r(\bb^*),
\end{array}$$
which contradicts the optimality of $\bb^*$. (We have used the absolute homogeneity of the norm $r$ and that $\bb^*\neq\mb0$.) Hence, any optimal $\bb^*$ to \eqref{eqn:pfsupp3} necessarily satisfies $\|\y-\X\bb^*\|_2 - \lambda\sum_{i=1}^k |\beta_{(i)}^*| \geq 0$ and so the desired results follows.
\end{proof}

\emph{N.B.} The assumption that $r$ is a norm can be relaxed somewhat (as is clear in the proof), although the full generality is not necessary for our purposes.

\subsection*{Corollary \ref{cor:slope} and related discussions}

Here we include a precise statement of the ``approximate'' claim in Corollary \ref{cor:slope}. After the proof, we include a discussion of related technical issues.%Let us set a standard notion: we say that $\bb$ is $\epsilon$-optimal (for $\epsilon>0$)  to an optimization problem $(\textrm{P})$ if the optimal objective value of $(\textrm{P})$ is within $\epsilon$ of the objective value of $\bb$.

\begin{theorem}[Precise statement of Corollary \ref{cor:slope}]\label{thm:corprecise}
For $\tau>\lambda>0$, consider the problems
\begin{equation}\label{eqn:ced1}% "corollary extended discussion 1"
\begin{array}{ll}
\ds\min_{\bb } &\ds\|\y-\X\bb\|_2 + (\tau-\lambda)\|\bb\|_1+ \lambda \tk{\bb} \\
\st & \ds\lambda\sum_{i=1}^k|\beta_{(i)}| \leq \|\y-\X\bb\|_2.
\end{array}
\end{equation}
and
\begin{equation}\label{eqn:ced2}% "corollary extended discussion 2"
\min_{\bb } \ds\|\y-\X\bb\|_2 + (\tau-\lambda)\|\bb\|_1+ \lambda \tk{\bb}.
\end{equation}
For all $\epsilon>0$, there exists $\barl=\barl(\epsilon)>0$ so that whenever $\lambda\in(0,\barl)$,
\begin{enumerate}
\item Every optimal $\bb^*$ to \eqref{eqn:ced1} is $\epsilon$-optimal to \eqref{eqn:ced2}.
\item For every optimal $\bb^*$ to \eqref{eqn:ced2}, there is some $\whbb$ so that $\|\bb^*-\whbb\|_2\leq \epsilon$, $\whbb$ is feasible to \eqref{eqn:ced1}, and $\whbb$ is $\epsilon$-optimal to \eqref{eqn:ced1}. 
\end{enumerate}
\end{theorem}

\begin{proof}
Fix $\tau>0$ throughout. We assume without loss of generality that $\y\neq\mb0$, as otherwise the claim is obvious. We will prove the second claim first, as it essentially implies the first.

Let us consider two situations. In particular, we consider whether there exists a nonzero optimal solution to
\[\label{eqn:ced3}
\min_\bb \|\y-\X\bb\|_2+\tau\|\bb\|_1.
\]

\subsubsection*{Case 1---existence of nonzero optimal solution to \eqref{eqn:ced3}}
We first consider the case when there exists a nonzero solution to problem \eqref{eqn:ced3}. We show a few lemmata:

\begin{enumerate}
\item We first show that the norm of solutions to \eqref{eqn:ced2} are uniformly bounded away from zero, independent of $\lambda$. To proceed,
let $\whbb$ be any nonzero optimal solution to \eqref{eqn:ced3}. Observe that if $\bbs$ is optimal to \eqref{eqn:ced2}, then
\begin{align*}
\|\y-\X\bbs\|_2 + (\tau-\lambda)\|\bbs\|_1 + \lambda T_k({\bbs}) & 
\leq \|\y-\X\whbb\|_2 + (\tau-\lambda)\|\whbb\|_1 + \lambda T_k(\whbb)\\
&\leq \|\y-\X\bbs\|_2 + \tau\|\bbs\|_1 - \lambda \|\whbb\|_1 + \lambda T_k(\whbb),
\end{align*}
implying that $\|\whbb\|_1 - T_k(\whbb) \leq \|\bbs\|_1 - T_k(\bbs)$. In other words, $\sum_{i=1}^k|\widehat{\beta}_{(i)}| \leq \sum_{i=1}^k |{\beta}_{(i)}^*|\leq \|\bbs\|_1$. Using the fact that $\whbb\neq\mb0$, we have that any solution $\bbs$ to \eqref{eqn:ced2} has strictly positive norm:
$$\|\bbs\|_1 \geq C>0,$$
where $C:=\sum_{i=1}^k|\widehat{\beta}_{(i)}|$ is a universal constant depending only on $\tau$ (and not $\lambda$).

\item We now upper bound the norm of solutions to \eqref{eqn:ced2}.  In particular, if $\bbs$ is optimal to \eqref{eqn:ced2}, then
$$\|\y-\X\bbs\|_2 + (\tau-\lambda)\|\bbs\|_1+ \lambda T_k({\bbs}) \leq \|\y\|_2 + 0 + 0 = \|\y\|_2,$$
and so $\|\bbs\|_1\leq \|\y\|_2/(\tau-\lambda)$. (This bound is not uniform in $\lambda$, but if we restrict our attention to, say $\lambda\leq \tau/2$, it is.)

\item We now lower bound the loss for scaled version of optimal solutions. In particular, if $\sigma\in[0,1]$ and $\bbs$ is optimal to \eqref{eqn:ced2}, then by optimality we have that
$$\|\y-\X\bbs\|_2 + (\tau-\lambda)\|\bbs\|_1 + \lambda T_k(\bbs) \leq \|\y-\sigma\X\bbs\|_2 + (\tau-\lambda)\sigma\|\bbs\|_1 + \lambda \sigma T_k(\bbs),$$
which in turn implies that
\begin{align*}
\|\y-\sigma\X\bbs\|_2 &\geq \|\y-\X\bbs\|_2 + (\tau-\lambda)(1-\sigma) \|\bbs\|_1 + \lambda (1-\sigma) T_k(\bbs)\\
& \geq \|\y-\X\bbs\|_2 + (\tau-\lambda)(1-\sigma) C\geq (\tau-\lambda)(1-\sigma)C
\end{align*}
by combining with the first observation.

\end{enumerate}

Using these, we are now ready to proceed. Let $\epsilon>0$; we assume without loss of generality that $\epsilon<2\|\y\|_2/\tau$. Let
$$\barl:= \min\left\{\frac{\epsilon\tau^3 C}{4\|\y\|_2(2\|\y\|_2 - \epsilon\tau)},\frac{\tau}{2}\right\}.$$
Fix $\lambda\in(0,\barl)$ and let $\bbs$ be any optimal solution to \eqref{eqn:ced2}.   Define 
$$\sigma:= \left(1-\frac{\epsilon\tau}{2\|\y\|_2} \right) \text{\quad and \quad}\whbb := \sigma\bbs.$$
We claim that $\whbb$ satisfies the desired requirements of the theorem:

\begin{enumerate}
\item We first argue that $\|\bbs-\whbb\|_2 \leq \epsilon$. Observe that
%\begin{align*}
%\|\bbs-\whbb\|_2 &= \epsilon\tau\|\bbs\|_2/({2\|\y\|_2} )\\
%& \leq {\epsilon\tau}\|\bbs\|_1/({2\|\y\|_2} )\\
%& \leq {\epsilon\tau} \|\y\|_2/({2\|\y\|_2} \tau)\\
%& = \epsilon\|\y\|_2/2.
%\end{align*}
$$\|\bbs-\whbb\|_2 = \epsilon\tau\|\bbs\|_2/({2\|\y\|_2} ) \leq {\epsilon\tau}\|\bbs\|_1/({2\|\y\|_2} ) \leq {\epsilon\tau} \|\y\|_2/({2\|\y\|_2}( \tau-\lambda)) \leq\epsilon.
$$

\item We now show that $\whbb$ is feasible to \eqref{eqn:ced1}. This requires us to argue that
$\lambda \sum_{i=1}^k |\widehat{\beta}_{(i)}| \leq \|\y-\X\whbb\|_2$. Yet,
\begin{align*}
\lambda\sum_{i=1}^k |\widehat{\beta}_{(i)}| &\leq \lambda \|\whbb\|_1 = \lambda \sigma\|\bbs\|_1\leq  2\lambda\sigma\|\y\|_2/\tau\leq  \frac{\tau}{2} (1-\sigma) C\\
& \leq (\tau-\lambda) (1-\sigma)C \leq \|\y-\sigma\X\bbs\|_2 = \|\y-\X\whbb\|_2,
\end{align*}
as desired. The only non-obvious step is the inequality $2\lambda\sigma\|\y\|_2/\tau\leq \tau(1-\sigma)C/2$, which follows from algebraic manipulations using the definitions of $\sigma$ and $\barl$.

\item Finally, we show that $\whbb$ is $\left(\epsilon\|\X\|_2\right)$-optimal to \eqref{eqn:ced1}. Indeed, because $\bbs$ is optimal to \eqref{eqn:ced2} which necessarily lowers bound problem \eqref{eqn:ced1}, we have that the objective value gap between $\whbb$ and an optimal solution to \eqref{eqn:ced1} is at most
\begin{align*}
&\|\y-\sigma\X\bbs\|_2 -\|\y-\X\bbs\|_2 + (\tau-\lambda)(\sigma-1)\|\bbs\|_1 + \lambda(\sigma-1)T_k(\bbs)\\
&\leq (1-\sigma)\|\X\bbs\|_2 + 0 + 0 \leq (1-\sigma)\|\X\|_2\|\bbs\|_2 \leq 2(1-\sigma) \|\X\|_2\|\y\|_2/\tau   \\
&=2\epsilon\tau/(2\|\y\|_2)\|\X\|_2\|\y\|_2/\tau = \epsilon\|\X\|_2.
\end{align*}
\end{enumerate}
As the choice of $\epsilon>0$ was arbitrary, this completes the proof of claim 2 in the theorem in the case when $\mb0$ is not a solution to \eqref{eqn:ced3}.

\subsubsection*{Case 2---no nonzero optimal solution to \eqref{eqn:ced3}}

In the case when there is no nonzero optimal solution to \eqref{eqn:ced3}, $\mb 0$ is optimal and it is the only optimal point. Our analysis will be similar to the previous approach, with the key difference being in how we lower bound the quantity $\|\y-\sigma\X\bbs\|_2$ where $\bbs$ is optimal to \eqref{eqn:ced2}. Again, we have several lemmata:

\begin{enumerate}
\item As before, if $\bbs$ is optimal to \eqref{eqn:ced2}, then $\|\bbs\|_1\leq \|\y\|_2/(\tau-\lambda)$.

\item We now lower bound the quantity $\|\y-\sigma\X\bbs\|_2$, where $\bbs$ is optimal to \eqref{eqn:ced2} and $\sigma\in[0,1]$. As such, consider the function
$$f(\sigma):= \|\y-\sigma\X\bbs\|_2 + \sigma\tau\|\bbs\|_1.$$
Because $f$ is convex in $\sigma$ and the unique optimal solution to \eqref{eqn:ced3} is $\mb0$, we have that
$$f(\sigma)\geq f(0) + \sigma f'(0)\;\;\;\forall\sigma\in[0,1]\text{\quad and \quad} f'(0)\geq0$$
(It is not difficult to argue that $f$ is differentiable at $0$.) An elementary computation shows that $f'(0) = \tau\|\bbs\|_1 -\langle\y,\X\bbs\rangle/\|\y\|_2$. Therefore, we have that
$$\|\y-\sigma\X\bbs\|_2 + \sigma\tau\|\bbs\|_1 \geq \|\y\|_2 + \sigma \left(\tau\|\bbs\|_1 - \langle \y,\X\bbs\rangle/\|\y\|_2\right),$$
implying that
$$\|\y-\sigma\X\bbs\|_2 \geq \|\y\|_2 - \sigma \langle \y,\X\bbs\rangle/\|\y\|_2 \geq \|\y\|_2 - \sigma \tau\|\bbs\|_1\geq \|\y\|_2 - \sigma\tau\|\y\|_2/(\tau-\lambda),$$
with the final step following by an application of the previous lemma.

\end{enumerate}

We are now ready to proceed. Let $\epsilon>0$; we assume without loss of generality that $\epsilon<2\|\y\|_2/\tau$. Let
$$\barl:= \min\left\{\frac{\epsilon\tau^2}{4\|\y\|_2 - \epsilon\tau},\frac{\tau}{2}\right\}.$$
Fix $\lambda\in(0,\barl)$ and let $\bbs$ be any optimal solution to \eqref{eqn:ced2}.   Define 
$$\sigma:= \left(1-\frac{\epsilon\tau}{2\|\y\|_2} \right) \text{\quad and \quad}\whbb := \sigma\bbs.$$
We claim that $\whbb$ satisfies the desired requirements:

\begin{enumerate}
\item The proof of the claim that $\|\bbs-\whbb\|_2 \leq \epsilon$ is exactly as before.

\item We now show that $\whbb$ is feasible to \eqref{eqn:ced1}, which requires a different proof. Again this requires us to argue that
$\lambda \sum_{i=1}^k |\widehat{\beta}_{(i)}| \leq \|\y-\X\whbb\|_2$. Yet,
\begin{align*}
\lambda\sum_{i=1}^k |\widehat{\beta}_{(i)}| &\leq \lambda \|\whbb\|_1 = \lambda \sigma\|\bbs\|_1\leq  \lambda\sigma\|\y\|_2/(\tau-\lambda)\leq \|\y\|_2 - \sigma\tau\|\y\|_2/(\tau-\lambda)\\
&\leq \|\y-\sigma\X\bbs\|_2 = \|\y-\X\whbb\|_2,
\end{align*}
as desired. The only non-obvious step is the inequality $\lambda \sigma\|\y\|_2/(\tau-\lambda) \leq \|\y\|_2 - \sigma\tau\|\y\|_2/(\tau-\lambda)$, which follows from algebraic manipulations using the definitions of $\sigma$ and $\barl$.

\item Finally, the proof that $\whbb$ is $\left(\epsilon\|\X\|_2\right)$-optimal to \eqref{eqn:ced1} follows in the same way as before.
\end{enumerate}

Therefore, we conclude that in the case when $\mb 0$ is the unique optimal solution to \eqref{eqn:ced3}, then again we have that the claim 2 of the theorem holds.

Finally, we show that claim 1 holds: any solution $\bbs$ to \eqref{eqn:ced1} is $\epsilon$-optimal to \eqref{eqn:ced2}. This follows by letting $\overline{\bb}$ be any optimal solution to \eqref{eqn:ced2}. By applying the entire argument above, we know that the objective value of some $\whbb$, feasible to \eqref{eqn:ced1} and close to $\overline{\bb}$, is within $\epsilon$ of the optimal objective value of \eqref{eqn:ced1}, i.e., the objective value of $\bbs$, and within $\epsilon$ of the objective value of \eqref{eqn:ced2}, i.e., the objective value of $\overline{\bb}$. This completes the proof.
\end{proof}

In short, the key complication is that the quantity $\|\y-\X\bb^*\|_2$ does not need to be uniformly bounded away from zero for solutions $\bb^*$ to problem \eqref{eqn:ced2}. This is part of the complication of working with the homogeneous form of the trimmed Lasso problem. For a concrete example, if one considers the homogeneous Lasso problem with $p=n=1$, $\y=(1)$, and $\X=(1)$, then the homogeneous Lasso problem $\min_\bb \|\y-\X\bb\|_2+\eta\|\bb\|_1$ is
$$\min_{\beta} |1-\beta| + \eta|\beta|.$$
For $\eta\in[0,1]$, $\beta^*=1$ is an optimal solution to this problem with corresponding error $\|\y-\X\bb^*\|=0$. If we make an assumption about the behavior of $\|\y-\X\bb^*\|$, then we do not need the setup as shown above.

\begin{proof}[Proof of Proposition \ref{prop:robeivslope}]
The proof is entirely analogous to that of Theorems \ref{thm:robeivInterp} and \ref{thm:corprecise} and is therefore omitted.
\end{proof}

\begin{proof}[Proof of validity of Example \ref{eg:cl}]

Let us consider the problem instance where $p=n=2$ with
$$\y = \begin{pmatrix}1\\1\end{pmatrix} \text{ \quad and \quad } \mb X = \begin{pmatrix} 1 & -1\\-1&2\end{pmatrix}.$$
Let $\lambda =1/2$ and $ \ell = 1$, and consider the problem
\begin{equation}\label{eqn:as1}
\min_{\bb} \|\y-\X\bb\|_2^2 + |\beta_{(2)}| = \min_{\beta_1,\beta_2} (1-\beta_1+\beta_2)^2 + (1+\beta_1-2\beta_2)^2 + |\beta_{(2)}|.
\end{equation}
We have omitted the factor of $1/2$ as shown in the actual example in the main text in order to avoid unnecessary complications.

Solving problem \eqref{eqn:as1} and its related counterparts (for $\ell\in\{0,2\}$) can rely on convex analysis because we can simply enumerate all possible scenarios. In particular, the solution to \eqref{eqn:as1} is $\bb^*=(3/2,1)$ based on an analysis of two related problems:
\begin{align*}
\min_{\beta_1,\beta_2} (1-\beta_1+\beta_2)^2 + (1+\beta_1-2\beta_2)^2 + |\beta_1|.\\
\min_{\beta_1,\beta_2} (1-\beta_1+\beta_2)^2 + (1+\beta_1-2\beta_2)^2 + |\beta_2|.
\end{align*}
(We should be careful to impose the additional constraints $|\beta_1|\leq |\beta_2|$ and $|\beta_1|\geq|\beta_2|$, respectively, although a simple argument shows that these constraints are not required in this example.) A standard convex analysis using the Lasso (e.g. by directly using subdifferentials) shows that the problems have respective solutions $(1/2,1/2)$ and $(3/2,1)$, with the latter having the better objective value in \eqref{eqn:as1}. As such, $\bb^*$ is indeed optimal. The solution in the cases of $\ell\in\{0,2\}$ follows a similarly standard analysis.

It is perhaps more interesting to study the general case where $\mu,\gamma\geq0$. In particular, we will show that $\bb^*=(3/2,1)$ is not an optimal solution to the clipped Lasso problem
\begin{equation}\label{eqn:as2}
\min_{\beta_1,\beta_2} (1-\beta_1+\beta_2)^2 + (1+\beta_1-2\beta_2)^2 + \mu\min\{\gamma|\beta_1|,1\}+\mu\min\{\gamma|\beta_2|,1\}
\end{equation}
for any choices of $\mu$ and $\gamma$. While in general such a problem may be difficult to fully analyze, we can again rely on localized analysis using convex analysis. To proceed, let
$$f(\beta_1,\beta_2) = (1-\beta_1+\beta_2)^2 + (1+\beta_1-2\beta_2)^2 + \mu\min\{\gamma|\beta_1|,1\}+\mu\min\{\gamma|\beta_2|,1\},$$
with the parameters $\mu$ and $\gamma$ implicit. We consider the following exhaustive cases:

\begin{enumerate}
\item $\boxed{\gamma>1}$ : In this case, $f$ is convex and differentiable in a neighborhood of $\bb^*$. Its gradient at $\bb^*$ is $\nabla f(\bb^*)=(0,-1)$, and therefore $\bb^*$ is neither locally optimal nor globally optimal to problem \eqref{eqn:as2}.

\item $\boxed{\gamma<2/3}$ : In this case, $f$ is again convex and differentiable in a neighborhood of $\bb^*$. Its gradient at $\bb^*$ is $\nabla f(\bb^*)=(\mu\gamma,\mu\gamma-1)$. Again, this cannot equal $(0,0)$ and therefore $\bb^*$ is neither locally nor globally optimal to problem \eqref{eqn:as2}.

\item $\boxed{2/3<\gamma<1}$ : In this case, $f$ is again convex and differentiable in a neighborhood of $\bb^*$. Its gradient at $\bb^*$ is $\nabla f(\bb^*)=(0,\mu\gamma-1)$. As a necessary condition for local optimality, we must have that $\mu\gamma=1$, implying that $\mu>1$. Further, if $\bb^*$ is optimal to \eqref{eqn:as2}, then $f(\bb^*)\leq f(0,0)$. Yet,
\begin{align*}
f(\bb^*) &= 1/2 + \mu + \mu\gamma = 3/2 + \mu\\
f(0,0) & = 2 ,
\end{align*}
implying that $\mu\leq 1/2 $, in contradiction of $\mu>1$. Hence, $\bb^*$ cannot be optimal to \eqref{eqn:as2}.

\item $\boxed{\gamma=2/3}$ : In this case, we make two comparisons, using the points $\bb^*$, $(0,0)$, and $(3,2)$:
\begin{align*}
f(\bb^*) &= 1/2 + \mu + 2\mu/3 = 1/2 + 5\mu/3\\
f(0,0) & = 2\\
f(3,2)&= 2\mu.
\end{align*}
Assuming optimality of $\bb^*$, we have that $f(\bb^*)\leq f(0,0)$, i.e., $\mu\leq 9/10$; similarly, $f(\bb^*)\leq f(3,2)$, i.e., $\mu\geq3/2$. Clearly both cannot hold, and so therefore $\bb^*$ cannot be optimal.

\item $\boxed{\gamma=1}$ : Finally, we see that $f(\bb^*)\leq f(3,2)$ would imply that $1/2+2\mu\leq 2\mu$, which is impossible; hence, $\bb^*$ is not optimal to \eqref{eqn:as2}. (This argument can clearly also be used in the case when $\gamma>1$, although it is instructive to see the argument given above in that case.)

\end{enumerate}

\noindent In any case, we have that $\bb^*$ cannot be a solution to the clipped Lasso problem \eqref{eqn:as2}. This completes the proof of validity of Example \ref{eg:cl}.
\end{proof}

\section{Supplementary details for Algorithms}\label{app:algsupp}

This appendix contains further details on algorithms as discussed in Section \ref{sec:algs}. The presentation here is primarily self-contained. Note that the alternating minimization scheme based on difference-of-convex optimization can be found in \cite{gotoh1}.

\subsection{Alternating minimization scheme}
 
Let us set the following notation:%We assume throughout that $k$, $\X$, $\y$, $\mu$, and $\lambda$ are fixed. Let
$$\begin{array}{lll}
f(\bb) &=& \|\y-\X\bb\|_2^2/2 + \lambda \tk{\bb}+ \eta \|\bb\|_1,\\
f_1(\bb) &=& \|\y-\X\bb\|_2^2/2 + (\eta+\lambda) \|\bb\|_1,\\
f_2(\bb) &=& \lambda \sum_{i=1}^k |\beta_{(i)}|.
\end{array}
$$

\begin{definition}
For any function $F:\R^p\to\R$ and $\epsilon\geq0$, we define the $\epsilon$-subdifferential of $F$ at $\bb_0\in\R^p$ to be the set $\partial_\epsilon F(\bb_0)$ defined as
$$\left\{ \bg\in\R^p\; : \; F(\bb) - F(\bb_0) \geq \langle\bg,\bb-\bb_0\rangle - \epsilon \;\forall \;\bb\in\R^p\right\}.$$
In particular, when $\epsilon=0$, we refer to $\sd_0 F(\bb_0)$ as the subdifferential of $F$ at $\bb_0$, and we will denote this as $\sd F(\bb_0)$.
\end{definition}

Using this definition, we have the following result precisely characterizing local and global optima of \eqref{eqn:alg}.

\begin{theorem}\label{thm:optCharac}
\begin{enumerate}[(a)]
\item A point $\bb^*$ is a local minimum of $f$ if and only if $\sd f_2(\bb^*) \sub \sd f_1(\bb^*)$.

\item A point $\bb^*$ is a global minimum of $f$ if and only if $\sd_\epsilon f_2(\bb^*) \sub \sd_\epsilon f_1(\bb^*)$ for all $\epsilon\geq0$.
\end{enumerate}
\end{theorem}
\begin{proof}
This is a direct application of results in \cite[Thm. 1]{taoan97}. Part (b) is immediate. The forward implication of part (a) is immediate as well; the converse implication follows by observing that $f_2$ is a \emph{polyhedral} convex function \cite[Thm. 1(ii)]{dcSummary} (see definition therein).
\end{proof}

Let us note that $\sd f_1$ and $\sd f_2$ are both easily computable, and hence,  local optimality can be verified given some candidate $\bb^*$ per Theorem \ref{thm:optCharac}.\footnote{For the specific functions of interest, verifying local optimality of a candidate $\bb^*$ can be performed in $O(p\min\{n,p\}+p\log p)$ operations; the first component relates to the computation of $\X'\X\bb^*$, while the second captures the sorting of the entries of $\bb^*$. See Appendix \ref{app:alg1supp} for details.
%However, $\sd_\epsilon f_1$ and $\sd_\epsilon f_2$ are not efficiently computable in general (assuming P$\neq$NP).
}
We now discuss the associated alternating minimization scheme (or equivalently, as a sequential linearization scheme), shown in Algorithm \ref{alg:1} for finding local optima of \eqref{eqn:alg} by making use of Theorem \ref{thm:optCharac}. Through what follows, we make use of the standard notion of a conjugate function, defined as follows:

\begin{definition}
For any function $F:\R^p\to\R$, we define its conjugate function $F^*:\R^p\to\R$ to be the function
$$F^*(\bg) = \sup_{\bb} \langle\bg,\bb - F(\bb)\rangle.$$
\end{definition}

We will make the following minor technical assumption: in step 2) of Algorithm \ref{alg:1}, we assume without loss of generality that the $\bg^\ell$ so computed satisfies the additional criteria:
\begin{enumerate}
\item it is an extreme point of the relevant feasible region, 
\item and that if $\sd f_2(\bb^\ell) \not\sub \sd f_1(\bb^\ell)$, then $\bg^\ell$ is chosen such that $\bg^\ell \in \sd f_2(\bb^\ell) \setminus\sd f_1(\bb^\ell)$.
\end{enumerate}
Solving \eqref{eqn:wrtbg} with these additional assumptions can nearly be solved in closed form by simply sorting the entries of $|\bb|$, i.e., by finding $|\beta_{(1)}|,\ldots,|\beta_{(p)}|$.%\footnote{Observe that the maximizer $\bg$ is unique iff $|\beta_{(k)}|>|\beta_{(k+1)}|$.} 
We must take some care to ensure that the second without loss of generality condition on $\bg$ is satisfied. This is straightforward but tedious; the details are shown in Appendix \ref{app:alg1supp}.

Using this modification, the convergence properties of Algorithm \ref{alg:1} can be proven as follows:

\begin{proof}[Proof of Theorem \ref{thm:altConvProp}]
This is an application of \cite[Thms. 3-5]{taoan97}. The only modification is in requiring that $\bg^\ell$ is chosen so that $\bg^\ell \in \sd f_2(\bb^*) \setminus\sd f_1(\bb^*)$ if $\bb^\ell$ is not a local minimum of $f$---see \cite[\textsection3.3]{taoan97} for a motivation and justification for such a modification. Finally, the correspondence between $\bg^\ell\in \sd f_2(\bb^\ell)$ and \eqref{eqn:wrtbg}, and between $\bb^{\ell+1}\in \sd f_1^*(\bg^\ell)$ and \eqref{eqn:wrtbb}, is clear from an elementary argument applied to subdifferentials of variational formulations of functions.
\end{proof}

\subsection{Algorithm \ref{alg:1}, Step 2}\label{app:alg1supp}

Here we present the details of solving \eqref{eqn:wrtbg} in Algorithm \ref{alg:1} in a way that ensures that the associated without loss of generality claims hold. In doing so, we also implicitly study how to verify the conditions for local optimality (\emph{c.f.} Theorem \ref{thm:optCharac}). 
Throughout, we use the $\sgn$ function defined as 
$$\sgn(x) = \left\{\begin{array}{rl}
1,& x>0\\
-1,&x<0\\
0,&x=0.
\end{array}\right.$$

For fixed $\bb$, the problem of interest is
\begin{equation*}%\label{eqn:wrtbg}
\begin{array}{ll}
\ds\max_{\bg} & \langle\bb,\bg\rangle\\
\st& \ds\sum_i |\gamma_i| \leq \lambda k\\
& \ds|\gamma_i|\leq \lambda\;\forall i.
\end{array}
\end{equation*}
We wish to find a maximizer $\bg$ for which the following hold:
\begin{enumerate}
\item $\bg$ is an extreme point of the relevant feasible region, 
\item and that if $\sd f_2(\bb) \not\sub \sd f_1(\bb)$, then $\bg$ is such that $\bg \in \sd f_2(\bb) \setminus\sd f_1(\bb)$.
\end{enumerate}
As the problem on its own can be solved by sorting the entries of $\bb$, the crux of the problem is ensuring that 2) holds.

Given the highly structured nature of $f_1$ and $f_2$ in our setup, it is simple, albeit tedious, to ensure that such a condition is satisfied. Let $I = \{i: |\beta_i|=|\beta_{(k)}|\}$. If $|I|=1$, the optimal solution is unique, and there is nothing to show. Therefore, we will assume that $|I|\geq2$. We will construct an optimal solution $\bg$ which satisfies the desired conditions. First observe that we necessarily must have that 1) $\gamma_i = \lambda\sgn(\beta_i)$ if $|\beta_i|> |\beta_{(k)}|$ and 2) $\gamma_i=0$ if $|\beta_i|<|\beta_{(k)}|$. We now proceed to define the rest of the entries of $\bg$. We consider two cases:

\begin{enumerate}
\item First consider the case when $|\beta_{(k)}|>0$. We claim that $\sd f_2(\bb) \not\sub \sd f_1(\bb)$. To do so, we will inspect the $i$th entries of $\sd f_1(\bb)$ for $i\in I$; as such, let $P_i^j =\{\delta_i: \bd\in \sd f_j(\bb)\}$ for $j\in\{1,2\}$ and $i\in I$ (a projection). For each $i\in I$, we have using basic convex analysis that $P_i^1$ is a singelton: $P_i^1 = \{\langle\X_i,\X\bb-\y\rangle + (\eta+\lambda)\sgn(\beta_i)\}$, where $\X_i$ is the $i$th column of $\X$. In contrast, because $|I|\geq2$, the set $P_i^2$ is an interval with strictly positive length for each $i\in I$ (it is either $[-\lambda,0]$ or $[0,\lambda]$, depending on whether $\beta_i<0$ or $\beta_i>0$, respectively). Therefore, $\sd f_2(\bb) \not\sub \sd f_1(\bb)$, as claimed.

Fix an arbitrary $j\in I$. Per the above argument, we must have that $\langle\X_j,\X\bb-\y\rangle+ (\eta+\lambda)\sgn(\beta_j)\neq 0$ or $\langle\X_j,\X\bb-\y\rangle + (\eta+\lambda)\sgn(\beta_j)\neq\lambda\sgn(\beta_j)$. In the former case, set $\gamma_i=0$, while in the latter case we define $\gamma_i=\lambda\sgn(\beta_i)$ (if both are true, either choice suffices). It is clear that it is possible to fill in the remaining entries of $\gamma_i$ for $i\in I\setminus\{j\}$ in a straightforward manner so that $\bg\in \partial f_2(\bb)$.  Further, by construction, $\bg\notin \partial f_1(\bb)$, as desired.

\item Now consider the case when $|\beta_{(k)}|=0$. Using the preceding argument, we see that $P_i^1$ is the interval $[\langle\X_i,\X\bb-\y\rangle-(\eta+\lambda),\langle\X_i,\X\bb-\y\rangle+\eta+\lambda] $ for $i\in I$. In contrast, $P_i^2$ is the interval $[-\lambda,\lambda]$ for $i\in I$. If for all $i\in I$ one has that $P_i^2\sub P_i^1$, then the choice of $\gamma_i$ for $i\in I$ is obvious: any optimal extreme point $\bg$ of the problem will suffice. (Note here that it may or may not be that $\sd f_2(\bb) \sub \sd f_1(\bb)$. This entirely depends on $\beta_i$ for $i\notin I$.)

Therefore, we may assume that there exists some $j\in I$ so that $P_j^2\not\sub P_j^1$. (It follows immediately that $\sd f_2(\bb) \not\sub \sd f_1(\bb)$.) We must have that $\langle\X_j,\X\bb-\y\rangle -(\eta+\lambda)>-\lambda$ or $\langle\X_j,\X\bb-\y\rangle + (\eta+\lambda)<\lambda$. In the former case, set $\gamma_i=-\lambda$, while in the latter case we define $\gamma_i=\lambda$ (if both are true, either choice suffices). It is clear that it is possible to fill in the remaining entries of $\gamma_i$ for $i\in I\setminus\{j\}$ in a straightforward manner so that $\bg\in \partial f_2(\bb)$. By construction, $\bg\notin \partial f_1(\bb)$, as desired.

\end{enumerate}

In either case, we have that one can choose $\bg\in\sd f_2(\bb)$ so that 1) $\bg$ is an extreme point of the feasible region $\{\bg:\sum_i|\gamma_i|\leq\lambda k,\; |\gamma_i|\leq \lambda\;\forall i\}$ and that 2) $\bg \in \sd f_2(\bb) \setminus\sd f_1(\bb)$ whenever $\sd f_2(\bb) \not\sub \sd f_1(\bb)$. This concludes the analysis; thus, we have shown the validity (and computational feasibility) of the without loss of generality claim present in Algorithm \ref{alg:1}. Indeed, per our analysis, Step 2 in Algorithm \ref{alg:1} can be solved in $O(p\min\{n,p\}+p\log p)$ operations (sorting of $\bb$ in $O(p\log p)$ followed by $O(p)$ conditionals and gradient evaluation in $O(np)$). In reality, if we keep track of gradients in Step 3, there is no need to recompute gradients in Step 2, and therefore in practice Step 2 is of the same complexity of sorting a list of $p$ numbers. (We assume that $\X'\y$ has been computed offline and store throughout for simplicity.)

\subsection{Algorithm \ref{alg:admm}, Step 3}\label{app:admmsupp}

Here we show how to solve Step 3 in Algorithm \ref{alg:admm}, namely, solving the orthogonal design trimmed Lasso problem
\begin{equation}\label{eqn:pf6}
\min_{\bg} \lambda \tk{\bg} + \frac{\sigma}{2}\|\bb-\bg\|_2^2  - \langle \mb q,\bg\rangle,
\end{equation}
where $\bb$ and $\mb q$ are fixed. This is solvable in closed form. Let $\ba=\bb-\mb q/\sigma$. First observe that we can rewrite \eqref{eqn:pf6} as
\begin{align*}
\eqref{eqn:pf6} &= \min_{\bg} \lambda \tk{\bg} + \sigma\|\bg-\ba\|_2^2/2\\
&= \min_{\substack{\bg,\zz:\\\sum_i z_i=p-k\\
\zz\in\{0,1\}^p}}\lambda\langle \zz,|\bg|\rangle + \sigma\|\bg-\ba\|_2^2/2\\
&= \min_{\substack{\bg,\zz:\\\sum_i z_i=p-k\\
\zz\in\{0,1\}^p}}\sum_i \left(\lambda z_i|\gamma| +   \sigma(\gamma_i-\alpha_i)^2/2\right).
\end{align*}
The penultimate step follows via Lemma \ref{lemma:miprep}. Per this final representation, the solution becomes clear. In particular, let $I$ be a set of $k$ indices of $\ba$ corresponding to $\alpha_{(1)}$, $\alpha_{(2)}$, \ldots, $\alpha_{(k)}$. (If $|\alpha_{(k)}| = |\alpha_{(k+1)}|$, we break ties arbitrarily.) Then a solution $\bg^*$ to \eqref{eqn:pf6} is
$$\gamma_i^* = \left\{\begin{array}{rl}
\alpha_i, & i\in I\\
\soft_{\lambda/\sigma}(\alpha_i),&i\notin I,
\end{array}\right.$$
where $\soft_{\lambda/\sigma}(\alpha_i) = \sgn(\alpha_i) \left|\alpha_i-\lambda/\sigma\right|$.

\subsection{Computational details}\label{app:compdetail}

For completeness and reproducibility, we also include all computational details. For Figure \ref{fig:coeffpath}, the following parameters were used to generate the test instance: $n = 100$, $p = 20$, $\text{SNR} = 10$, \texttt{julia} seed = 1, $\eta=0.01$, $k=2$. The example was generated from the following true model:
\begin{enumerate}
\item $\bb_\text{true}$ is a vector with ten entries equal to 1 and all others equal to zero. (So $\|\bb_\text{true} \|_0=10$.)
\item covariance matrix $\s$ is generated with $\Sigma_{ij} = .8^{|i-j|}$.
\item $\X\sim N(\mb 0,\s)$.
\item $\epsilon_i\stackrel{\text{i.i.d.}}{\sim} N(0, \bb_0'\s\bb_0/\text{SNR})$
\item $\y$ is then defined as $\X\bb_0+\be$

\end{enumerate}

The 100 examples generated for Figure \ref{fig:optgap} were using the following parameters: $n = 100$, $p = 20$, $\text{SNR} = 10$, \texttt{julia} seed $\in\{ 1,\ldots,100\}$, $\eta=0.01$, $k=2$, $\text{bigM} = 20$. MIO using Gurobi solver. Max iterations: alternating minimization---1000; ADMM (inner)---2000; ADMM (outer)---10000. ADMM parameters: $\sigma=1$, $\tau=0.9$. The examples themselves had the same structure as the previous example. The optimal gaps shown are relative to the objective in \eqref{eqn:alg}. The averages are computed as geometric means (relative to optimal 100\%) across the 100 instances, and then displayed relative to the optimal 100\%.

\end{appendices}

\section*{Acknowledgments}

Copenhaver was partially supported by the Department of Defense, Office of Naval Research, through the National Defense Science and Engineering Graduate (NDSEG) Fellowship. Mazumder was partially supported by ONR Grant N000141512342.

%%%%%%%%%%%%%%%%%%
%%% Bibliography                  %%%
%%%%%%%%%%%%%%%%%%

\bibliographystyle{IEEEtranS}
\bibliography{References-clean}

\end{document}